\documentclass[journal]{IEEEtran}
\usepackage[T1]{fontenc}

\ifCLASSINFOpdf
\else
\fi
\usepackage{amsmath}
\usepackage{graphicx}

\usepackage{color}
\usepackage{amsfonts}
\usepackage{amsmath}
\usepackage{amsthm}
\usepackage{amssymb}
\usepackage{alltt}
\usepackage{mathrsfs}
\usepackage{subfigure}
\usepackage{enumerate}
\usepackage[noadjust]{cite}
\usepackage{float}
\usepackage{hyperref}
\hypersetup{hidelinks}

\DeclareGraphicsExtensions{.pdf}
\usepackage{soul}
\usepackage{array}
\usepackage{bm}

\usepackage{algorithm}
\usepackage{algpseudocode}

\usepackage{booktabs}
\usepackage{multirow}
\usepackage{hhline}

\newtheorem{proposition}{Proposition}

\DeclareMathOperator*{\argmax}{arg\,max}

\newcommand{\mb}[1]{\mathbf{#1}}
\newcommand{\mc}[1]{\mathcal{#1}}
\newcommand{\mbb}[1]{\mathbb{#1}}
\newcommand{\mr}[1]{\mathrm{#1}}

\newlength{\figWidth}
\newmuskip\pFqmuskip

\newcommand*\pFq[6][8]{%
	\begingroup 
	\pFqmuskip=#1mu\relax
	\mathchardef\normalcomma=\mathcode`,
	\mathcode`\,=\string"8000
	\begingroup\lccode`\~=`\,
	\lowercase{\endgroup\let~}\pFqcomma
	{}_{#2}F_{#3}{\left[\genfrac..{0pt}{}{#4}{#5};#6\right]}%
	\endgroup
}
\newcommand{\pFqcomma}{{\normalcomma}\mskip\pFqmuskip}

\begin{document}
%
\title{Secure UAV-to-Ground MIMO Communications: Joint Transceiver and Location Optimization}
%
%
%

\author{Zhong~Zheng,~\IEEEmembership{Member,~IEEE}, Xinyao Wang, Zesong Fei,~\IEEEmembership{Senior Member,~IEEE}, Qingqing Wu,~\IEEEmembership{Member,~IEEE}, Bin Li,~\IEEEmembership{Member,~IEEE}, and Lajos Hanzo,~\IEEEmembership{Fellow, IEEE}
\thanks{Copyright (c) 2015 IEEE. Personal use of this material is permitted. However, permission to use this material for any other purposes must be obtained from the IEEE by sending a request to pubs-permissions@ieee.org.
}

\thanks{Z. Zheng, X. Wang, and Z. Fei are with the School of Information and Electronics, Beijing Institute of Technology, Beijing 100081, China. (e-mail: \{zhong.zheng, 3120205414, feizesong\}@bit.edu.cn).
	
Q. Wu is with the State Key Laboratory of Internet of Things for Smart City, University of Macau, Macau, 999078, China (e-mail: qingqingwu@um.edu.mo).

B. Li is with the School of Computer and Software, Nanjing University of Information Science and Technology, Nanjing 210044, China, and also with the Guangxi Key Laboratory of Multimedia Communications and Network Technology, Nanning 530004, China (e-mail: bin.li@nuist.edu.cn).

L. Hanzo is with the School of Electronics and Computer Science, University of Southampton, Southampton, Hants SO17 1BJ, U.K. (e-mail: lh@ecs.soton.ac.uk).}}

%
%

\markboth{}
{}
%



\maketitle

\begin{abstract}
Unmanned aerial vehicles~(UAVs) are foreseen to constitute promising airborne communication devices as a benefit of their superior channel quality. But UAV-to-ground (U2G) communications are vulnerable to eavesdropping. Hence, we conceive a sophisticated physical layer security solution for improving the secrecy rate of multi-antenna aided U2G systems. Explicitly, the secrecy rate of the U2G MIMO wiretap channels is derived by using random matrix theory. The resultant explicit expression is then applied in the joint optimization of the MIMO transceiver and the UAV location relying on an alternating optimization technique. Our numerical results show that the joint transceiver and location optimization conceived facilitates secure communications even in the challenging scenario, where the legitimate channel of confidential information is inferior to the eavesdropping channel.
\end{abstract}

\begin{IEEEkeywords}
UAV; physical layer security; MIMO; random matrix theory; alternating optimization.
\end{IEEEkeywords}

\IEEEpeerreviewmaketitle

\section{Introduction}

\subsection{Backgrounds}

With the emergence of novel communication paradigms and scenarios, communication networks are undergoing significant changes with new technologies and network components introduced~\cite{Hanzoinvited}. Among them, the family of unmanned aerial vehicles~(UAVs) are foreseen as integral parts of future wireless communication systems~\cite{You_6G}. Compared to the terrestrial infrastructure, UAV-mounted network nodes have the advantages of high flexibility and low roll out time in cases of prompt on-demand network deployment~\cite{ZhouCommMagazine2018}, such as the occasional large-scale rallies or sporting events~\cite{AlouiniVehMagazine2020}, disaster recovery and rescue~\cite{ChenTIFS2019}, emergency response~\cite{KumarTIFS2020}, etc. The UAV-mounted transceiver typically has a line-of-sight~(LOS) air-to-ground channel due to its altitude and flexible maneuverability~\cite{ZengCommMagazine2016}. However, UAVs are typically constrained to short-term deployments due to their limited on-board battery capacity that results in relatively short flight endurance~\cite{Hanzo_TVT2019}. Additionally, as the UAV hovers in the open sky while communicating with the ground nodes, it is vulnerable to eavesdroppers.

In order to prolong the UAV's flight time, we consider laser-charged UAVs as the payload carrying drones, where the UAVs can be wirelessly powered by a laser beam transmitter~\cite{LiuVTMag2016}. As reported in \cite{Achtelik_2011}, such wireless power transfer techniques are capable of supplying sufficient energy for a quadcopter drone for more than 12 hours of uninterrupted flight. Additionally, laser-charging has been considered in~\cite{LiuTVT2019,BaiTVT2019,LiuWCMag2021} for various UAV-aided communication systems in order to circumvent frequent battery recharges or replacements, thus reducing the service interruption of UAV drones. Currently, the effective laser-charging distance is limited to a few tens of meters with output power up to 200 Watts~\cite{JinTPE2019}, it is therefore applicable in similar scenarios as the tethered UAV~\cite{AlouiniVehMagazine2020} or tethered helium-filled balloon~\cite{AlouiniProcIEEE2020}. Moreover, laser-charged UAVs have superior load-bearing capabilities, while avoiding safety issues of tethered airborne platforms, such as tether tangling and malicious damage.

To enhance the information security of UAV-to-ground~(U2G) communications, we harness physical layer security techniques for protecting the information from eavesdropping, which potentially guarantees perfect secrecy in contrast to the classic cryptography-based approaches, regardless of the computing power of the malicious nodes~\cite{Wyner1975}. Throughout this paper, the UAV and the target node on the ground are referred to as the sources~($\mathcal{S}$) and the destination~($\mathcal{D}$), while the channel between them is the legitimate channel. By contrast, the channel between the source and the eavesdropper~($\mathcal{E}$) is referred to as the eavesdropping channel.

Fig.~\ref{figUseCase} highlights some typical use cases of UAV-based airborne systems, where a laser-charging UAV is used for on-site video surveillance of a highway, rural area, and disaster area. In these cases, uninterrupted surveillance is required for relatively long periods, while the conventional base-station~(BS) infrastructure cannot be promptly established. The laser-charged UAVs can be dispatched right at the point of interest and the video stream is then wirelessly transmitted to the data collector for further analysis. Note that the remote data collector could be deployed at a conveniently reachable location and may potentially support multiple UAVs for efficient data collection. The sensitive private information captured is protected by physical layer security techniques during its transmissions, and the UAV is able to adjust its position for improving the secrecy of its transmission to the data collector.

\begin{figure}[t!]
	\centering
	\includegraphics[width=3.3in]{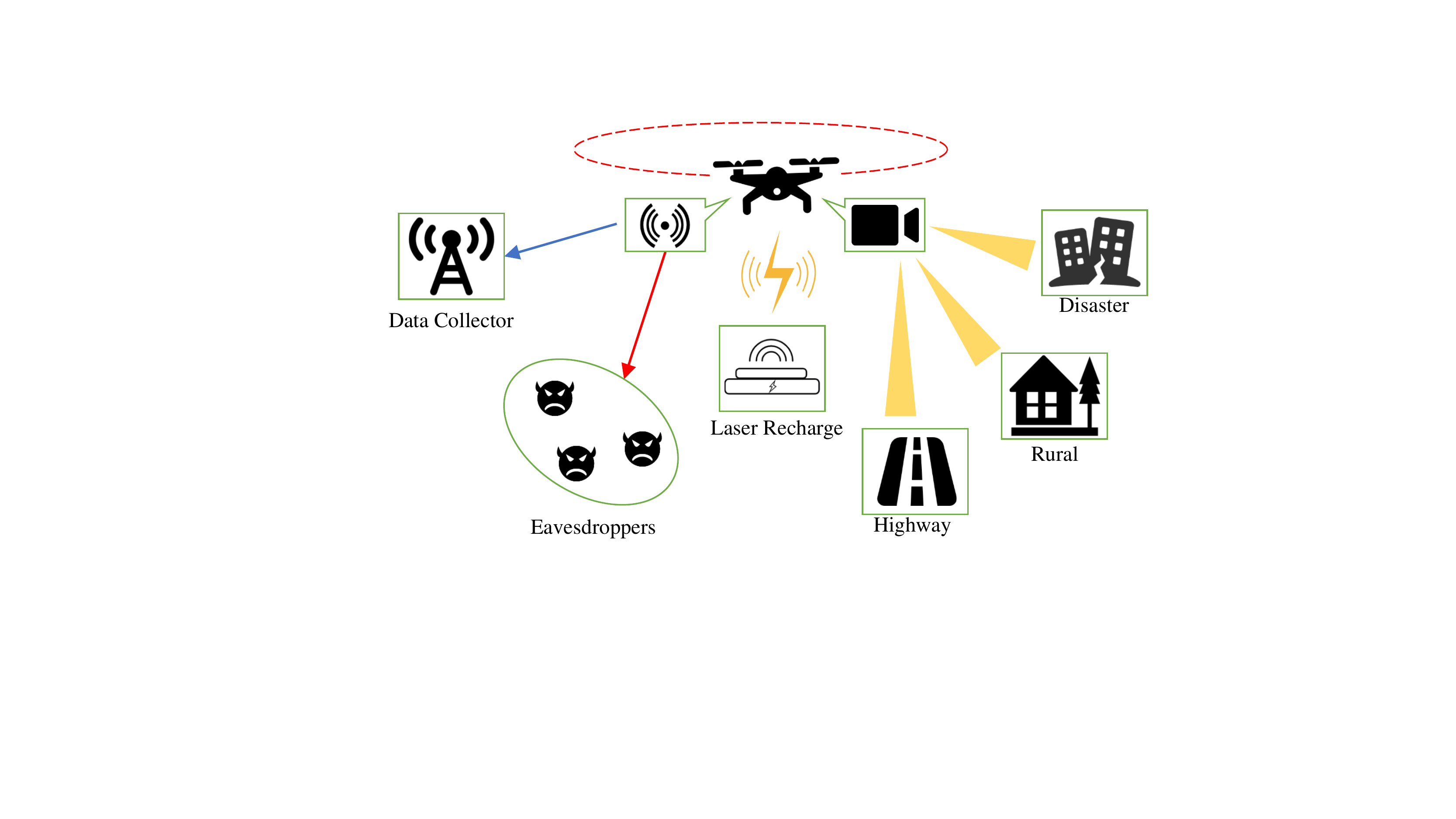}
	\caption{Laser-charging UAV in temporarily deployable video surveillance.}
	\label{figUseCase}
\end{figure}

\subsection{Related Works}
One of the main obstacles in the way of routinely relying on physical layer security is that the achievable information rate under the secrecy constraint heavily depends on the quality of the wireless channels. As the secrecy rate of the Gaussian wiretap channel is given by the difference of the Shannon rates between the legitimate and the eavesdropping channels~\cite{CsiszarTIT1978}, there may exist a secrecy outage zone, where no positive secrecy rate can be achieved. In order to support secrecy communications throughout the entire coverage area, the existing literature on secure UAV communications tends to advocate optimizing the transceiver for exploiting the mobility of the UAV to maximize its secrecy rate~\cite{WuWirelessMagazine2019}. In~\cite{KangCL2019}, the average worst case secrecy rate is maximized for a single-UAV network by optimizing the UAV's hovering altitude. The authors of~\cite{YaoTCOM2019} extend the secrecy rate optimization to the multi-UAV network, where the altitude of the UAV swarm can be adjusted. In contrast to the above-mentioned altitude-only optimization, joint power and unconstrained trajectory optimization is considered in~\cite{ZhangTWC2019,LiCL2019,DuoCC2021}, where the average secrecy rate within a certain flight duration is maximized by optimizing both the flight trajectory as well as the power allocation along the trajectory. Additionally, in secure UAV communications, the above-mentioned trajectory and power optimization has also been combined with user scheduling~\cite{LiTCOM2020}, artificial noise injection~\cite{XuTCOM2021}, friendly jamming~\cite{CaiJSAC2018}, etc.

Note that the authors of~\cite{KangCL2019,YaoTCOM2019,ZhangTWC2019,LiCL2019,DuoCC2021,LiTCOM2020,XuTCOM2021,CaiJSAC2018} only consider a single antenna at the UAVs, while designing the transmit power and deployment/trajectory. On the other hand, multi-antenna techniques are capable of substantially increasing the secrecy rate by exploiting the additional spatial Degree-of-Freedom (DoF) of Multi-Input Multi-Output (MIMO) channels~\cite{KhistiTIT2010}. However, the secrecy rate of MIMO-aided U2G communications has only been studied in a limited number of scenarios. In~\cite{WangAccess2019}, the authors consider cooperative MIMO transmission from an entire UAV swarm to the ground user under secrecy constraints. The UAVs in the swarm follow a predefined flight trajectory, while optimizing their transmit signals. In~\cite{LyuTVT2021}, a multi-antenna UAV applies zero-forcing transmit precoding~(TPC) for simultaneously communicating with multiple ground users under certain secrecy constraints, where the secrecy outage probability is obtained by modeling the positions of both the UAV and of the receivers as stochastic point processes. However, in~\cite{LyuTVT2021} and~\cite{WangAccess2019}, the inherent mobility of UAVs was not exploited for improving the security. In Table~\ref{tabCompare}, we boldly and explicitly contrast our novel contributions to the state-of-the-art.

In U2G MIMO communications, often there is both a dominant propagation component plus some additional random scattering components, where the theoretical model capturing these characteristics is typically the Rician MIMO channel~\cite{WillinkTVT2016}. Additionally, the eavesdroppers usually act as passive receivers and only their statistical channel state information (CSI), such as the mean and variance of the channel coefficients, may be known to $\mathcal{S}$. Under these assumptions, it is a challenge to analyze and optimize the secrecy rate of Rician wiretap channels, since no easy-to-manipulate mathematical expression of the information rate associated with arbitrary inputs is available, which would be required for the associated secrecy rate formulation. Previous results of secrecy rate optimization over Rician MIMO channels only exist for some particular scenarios assuming either uniform power allocation~\cite{AhmedAccess2018} or open-loop orthogonal space-time block coding~(OSTBC) transmission~\cite{FerdinandWCL2013}, or alternatively assuming Multi-Input Single-Output (MISO) transmission for the legitimate channel~\cite{LiTIFS2011}. Additionally, the information rates of the canonical Rician MIMO channels having arbitrary inputs have been characterized, for example in~\cite{McKayTIT2005,McKayCL2006,MatthaiouTCOM2009,TariccoTIT2008,ZhangJSAC2013}. In~\cite{McKayTIT2005,McKayCL2006,MatthaiouTCOM2009}, the Jensen and Minkowski type upper and lower bounds of the ergodic rate of Rician MIMO scenarios were derived, and some of the results rely on matrix-variate zonal polynomials, which are quite cumbersome in practice for numerical calculations. In~\cite{TariccoTIT2008} and~\cite{ZhangJSAC2013}, the authors resort to large matrix analysis for deriving the asymptotic ergodic rate of Rician MIMO scenarios, which become asymptotically tight when the number of antennas tends to infinity. However, this technique requires the solution of fixed-point equations for each transmitter and receiver pair. When the number of destinations and eavesdroppers is large, the computational complexity might become excessive for iterative optimization.

\begin{table*}[t]
	\renewcommand{\arraystretch}{1.2}%
	\centering
	\caption{\label{tabCompare} Recent Contributions on Secure UAV Communications}
	\begin{tabular}{|p{60pt}|p{90pt}|p{10pt}|p{10pt}|p{10pt}|p{10pt}|p{10pt}|p{10pt}|p{10pt}|p{10pt}|p{10pt}|p{10pt}|p{30pt}|}
		\hline
		\multicolumn{2}{|c|} {\textbf{Feature}} & \cite{KangCL2019} & \cite{YaoTCOM2019} & \cite{ZhangTWC2019} & \cite{LiCL2019} & \cite{DuoCC2021} & \cite{LiTCOM2020} & \cite{XuTCOM2021} & \cite{CaiJSAC2018} & \cite{LyuTVT2021} & \cite{WangAccess2019} & \text{Our Paper} \\
		\hhline{|=|=|=|=|=|=|=|=|=|=|=|=|=|}
		
		\multirow{3}{*}{Antenna config} & Multi-antenna $\mathcal{S}$ &       &      &      &      &      &      &      &      & $\surd$   &$\surd$  & $\surd$ \\
		\cline{2-13}
		\multirow{4}{*}{} & Multi-antenna $\mathcal{D}$       &      &      &      &      &      &      &      &      &    &    $\surd$     & $\surd$ \\
		\cline{2-13}
		\multirow{4}{*}{} & Multi-antenna $\mathcal{E}$       &      &      &      &      &      &      &     &      &  &   $\surd$   & $\surd$   \\
		\hhline{|=|=|=|=|=|=|=|=|=|=|=|=|=|}
		
		\multirow{3}{*}{Channel model} & Non-fading channel  & $\surd$   & $\surd$   & $\surd$   & $\surd$   & $\surd$    & $\surd$  & $\surd$  & $\surd$  &      &       &      \\
		\cline{2-13}
		\multirow{4}{*}{} & Rayleigh fading &      & $\surd$ & $\surd$  &       &$\surd$  &        &       &      & $\surd$  & $\surd$ &      \\
		\cline{2-13}
		\multirow{4}{*}{} & Rician fading  &       &      &      &       &        &       &       &       &       &       & $\surd$ \\
		\hhline{|=|=|=|=|=|=|=|=|=|=|=|=|=|}
		
		\multirow{2}{*}{CSI knowledge} & Full CSIs at UAV  & $\surd$ & $\surd$ & $\surd$ & $\surd$ & $\surd$ & $\surd$ & $\surd$ & $\surd$ &  $\surd$     &  &      \\
		\cline{2-13}
		\multirow{4}{*}{} & Statistical CSIs at UAV &       &        &       &          &        &            &         &        &   &   $\surd$      & $\surd$   \\
		\hhline{|=|=|=|=|=|=|=|=|=|=|=|=|=|}
		
		
		\multirow{2}{*}{Transceiver design} & Linear precoding &        &        &          &         &        &         &        &         &  $\surd$  & $\surd$ &      \\
		\cline{2-13}
		\multirow{4}{*}{} & Optimized precoding &        &        &          &         &        &         &      &        &       &        &  $\surd$  \\
		\hhline{|=|=|=|=|=|=|=|=|=|=|=|=|=|}
		
		\multirow{3}{*}{Location design} & Altitude only opt.& $\surd$ & $\surd$ &        &       &          &        &          &        &        &        &        \\
		\cline{2-13}
		\multirow{4}{*}{} & Unconstrained location opt.          &      &       &$\surd$ & $\surd$  &$\surd$  &       & $\surd$ & $\surd$  &       &       &     \\
		\cline{2-13}
		\multirow{4}{*}{} & Constrained location opt.      &      &       &        &         &        & $\surd$ &         &         &        &        &$\surd$     \\
		\hhline{|=|=|=|=|=|=|=|=|=|=|=|=|=|}
		
		\multirow{3}{*}{Secrecy metric} & Average secrecy rate &      &       & $\surd$ & $\surd$ &  $\surd$ &         & $\surd$  &       &  &  $\surd$     &      \\
		\cline{2-13}
		\multirow{4}{*}{} & Outage secrecy rate               &      &  $\surd$  &        &        &       &        &         &       &  $\surd$   &      &     \\
		\cline{2-13}
		\multirow{4}{*}{} & Worst case secrecy rate           &$\surd$ &      &        &        &        &$\surd$ &        &$\surd$ &       &  & $\surd$   \\
		\hline
	\end{tabular}
\end{table*}


\subsection{Contributions}
Since the existing results on U2G MIMO wiretap channels are either computationally complex or remain implicit in terms of the optimization variables, they cannot be directly applied for secrecy rate maximization. Therefore, we first construct a closed-form expression for the secrecy rate of generic U2G MIMO channels, when only statistical CSI is assumed at $\mathcal{S}$. The secrecy rate obtained is an explicit function of both the covariance matrix of the signals as well as of the UAV location, which readily lends itself to standard optimization techniques. The detailed contributions of this work are summarized as follows:
\begin{enumerate}
	\item We consider MIMO-aided transmissions from a laser-charged UAV to the ground node under a specific secrecy constraint. This scenario benefits both from the increased DoF of the signals and from the maneuverability of the UAV, each of which has only been applied for secure communications in isolation, but they have not been jointly optimized. We fill this knowledge-gap by jointly optimizing the transceiver and the UAV location. Explicitly, it will be shown in Section~\ref{secResult} that the joint optimization substantially improves the secrecy rate and avoids secrecy outages.
	\item In the secrecy rate construction, the Shannon rate of the legitimate channel is lower-bounded by Minkowski's inequality of the information rate of an equivalent lower dimensional MIMO channel, which avoids the cumbersome zonal polynomials. On the other hand, the Shannon rate of the eavesdropping channel is obtained by exploiting the random position of the eigenspaces of the legitimate and eavesdropping channel matrices, where the Haar matrix insertion technique of~\cite{ZhengTCOM2019} is applied.
	\item  In the secrecy rate optimization, an alternating optimization framework is proposed for designing the transmit signals and the deployment of the UAV for circumventing the non-convexity of the objective function. Explicitly, the signal-optimization subproblem and the deployment-optimization subproblem are solved alternatively, where each subproblem is approximated by its first-order Taylor series. This allows us to solve them efficiently by using standard optimization packages.
\end{enumerate}

The rest of the paper is organized as follows. Section~\ref{secModel} introduces the signal model of the U2G MIMO wiretap channels and outlines the secrecy rate optimization problem. In Section~\ref{secRate}, we characterize the information rate of both the legitimate and of the eavesdropping channels only relying on statistical channel knowledge at $\mathcal{S}$. In Section~\ref{secOpt}, we present the alternating signal and location optimization framework, while the numerical results are provided in Section~\ref{secResult}. Finally, Section~\ref{secConclude} concludes the paper.

\emph{Notations.} Throughout the paper, $\mb{A}$, $\bm{a}$, and $a$ denote matrices, vectors, and scalars, respectively. We denote $\mb{0}_{n}$ as the $n\times 1$ zero vector and $\mb{I}_n$ as the $n\times n$ identity matrix. The $(m\times n)$-element matrix block is defined as $\left[\left\{a_{i,j}\right\}_{m\times n}\right]$, where $a_{i,j}$ is the block element at the $i$-th row and $j$-th column, where the range of subscripts are $1\le i\le m$, $1\le j\le n$. The matrix conjugate-transpose, the matrix transpose, and the complex conjugate are denoted as $(\cdot)^\dagger$, $(\cdot)^\mathrm{T}$, and $(\cdot)^*$, respectively. The matrix trace and determinant are denoted as $\mathrm{Tr}(\mb{A})$ and $|\mb{A}|$. We use $\mc{CN}(\mb{0}_n,\mb{A})$ to denote the zero-mean complex Gaussian vector with covariance matrix $\mb{A}$ and $\mbb{E}\{\cdot\}$  is the expectation of a random variable.

\section{System Model}\label{secModel}
Consider the transmissions of a UAV to a ground node in the presence of $T$ non-colluding eavesdroppers. It is assumed that the UAV is equipped with $K$ antennas and the $\tau$-th receiver is equipped with $N_\tau$ antennas, where $0\le \tau\le T$. The receiver having the index of $\tau = 0$ refers to the destination, while the receivers with indices $1\le \tau\le T$ refer to eavesdroppers. The three-dimensional coordinates of the UAV are denoted as $\bm{p}_u = [p_{u,1}, p_{u,2}, p_{u,3}]^{\mathrm{T}}$, where $p_{u,1}$, $p_{u,2}$, and $p_{u,3}$ are the latitude, the longitude, and the altitude of the UAV, respectively. Similarly, the coordinates of the $\tau$-th receiver is denoted as $\bm{p}_\tau = [p_{\tau,1}, p_{\tau,2}, 0]^{\mathrm{T}}$. We assume that the UAV maintains a constant hovering altitude to avoid frequent lifting. Furthermore, the UAV is constrained to a certain maximum displacement $d_{\max}$ from the initial dispatch location, since the efficiency of power transfer drastically decreases for long laser-charging distances~\cite{LiuVTMag2016}.

\subsection{Signal Model}
Given the coordinates of the UAV as $\bm{p}_u$ and the transmit signal as $\bm{x} = [x_1,\ldots,x_K]^{\mathrm{T}}$, the signal received by the $\tau$-th receiver is denoted by $\bm{y}_\tau = [y_{\tau,1},\ldots,y_{\tau,N_\tau}]^{\mathrm{T}}$, which can be expressed as
\begin{align}
	\bm{y}_\tau = \sqrt{\rho_\tau(\bm{p}_u)} \mb{H}_\tau \bm{x} + \bm{n}_\tau,\quad 0\le \tau\le T,\label{eqy}
\end{align}
where $\rho_\tau(\bm{p}_u)$ denotes the large-scale propagation loss between the UAV and the $\tau$-th receiver, while the matrix $\mb{H}_\tau$ denotes the channel coefficients between the multi-antenna transmitter and receiver. The vector $\mb{n}_\tau$ is the additive white Gaussian noise, distributed as $\mb{n}_\tau\sim\mathcal{CN}(\mb{0}, \nu\mb{I})$.

The propagation loss depends on the propagation distance, which can be formulated as
\begin{equation}\label{eqRho}
\rho_\tau(\bm{p}_u) = \frac{c_p}{||\bm{p}_u - \bm{p}_\tau||^\alpha},
\end{equation}
where $c_p$ is the propagation loss at unit distance, $\alpha\ge 2$ is the path loss exponent, and $||\bm{p}_u - \bm{p}_\tau||$ is the Euclidean distance between the UAV and the $\tau$-th receiver. Given the UAV's hovering height, the channel coefficients $\mb{H}_\tau$ is verified in~[40, 41] to follow the Rician distribution in the U2G transmissions. Thus, we assume that the U2G channel $\mb{H}_\tau$ is modeled as
\begin{equation}\label{eqH}
	\mathbf{H}_{\tau}=\sqrt{\frac{\kappa_{\tau}}{\kappa_{\tau}+1}} \mathbf{H}_{d, \tau}+\sqrt{\frac{1}{\kappa_{\tau}+1}} \mathbf{H}_{s, \tau}, 
\end{equation}
where $\mathbf{H}_{d, \tau}$ denotes the deterministic Line-of-Sight~(LoS) propagation components, and $\mathbf{H}_{s, \tau}$ represents the random scattering due to multi-path fading. The ratio between the squared norm of the LoS and the scattering components is given by the Rician factor $\kappa_\tau$, which depends on the location of the $\tau$-th receiver, $0\le \tau\le T$. We assume that each entry of $\mb{H}_{s,\tau}$ is an independent and identically distributed (\emph{i.i.d.}) standard complex Gaussian random variable, i.e., we have $\left[\mb{H}_{s,\tau}\right]_{n,k}\sim\mc{CN}(0,1)$. Each entry of $\mb{H}_{d,\tau}$ is of unit modulus and accounts for the phase change due to the signal propagation between the antenna elements of the transmitter and the receiver, i.e., we have $\left[\mb{H}_{d,\tau}\right]_{n,k} = \mathrm{exp}\left(i\frac{2\pi}{\lambda}r_{\tau,n,k}^{1/2}\right)$, where $r_{\tau,n,k}$ is the distance between the $k$-th antenna of the transmitter and the $n$-th antenna of the $\tau$-th receiver, while $\lambda$ denotes the wavelength. Note that the LoS components $\mb{H}_{d,\tau}$ are determined by the angle-of-departure and angle-of-arrival of the transmit and receive antenna arrays, as well as by the configuration of the antenna elements~\cite{WONG1999}.

\subsection{Assumptions on System Model}
In this work, we adopt the following assumptions concerning the system model:
\begin{itemize}
\item The channel $\mb{H}_{\tau}$ is frequency-flat and obeys block-fading, i.e. the entries of $\mb{H}_\tau$ vary independently from one coherence time to another, but remain constant within each coherent interval;
\item The instantaneous CSI is known by the receiver, while the eavesdroppers additionally know the instantaneous CSI of $\mb{H}_0$. On the other hand, $\mathcal{S}$ only knows the statistical CSIs of $\mb{H}_\tau$, $0\le \tau \le T$;
\item The potential deployment locations of the UAV are within a certain distance $d_{\max}$ from the initial dispatch location $\bm{p}_u^{(1)}$, i.e., $|| \bm{p}_u - \bm{p}_u^{(1)} ||\le d_{\max}$. The distance $d_{\max}$ is determined by the maximum charging range of the laser transmitter.
\end{itemize}

Due to the hardware and signaling limitations of the transceiver, perfect instantaneous CSI cannot be acquired at $\mathcal{S}$. Additionally, the eavesdroppers typically act as receivers only, therefore, $\mathcal{S}$ cannot obtain their instantaneous CSI in practice. However, thanks to the recent development of sophisticated sensing and detection capabilities for the UAV~\cite{LiuTCOM2020}, they are capable of detecting the distance and angle of the surrounding unidentified communication devices. Therefore, the distance- and angle-dependent channel parameters, such as the propagation loss $\rho_\tau(\bm{p}_u)$ and the deterministic LoS components $\mb{H}_{d,\tau}$, can be detected and leveraged by $\mathcal{S}$ as \emph{a priori} information in support of optimal signal and location designs.

\subsection{Secrecy Rate via Joint Transceiver and Location Optimizations}
The U2G wiretap channel considered in (\ref{eqy}) in the presence of multiple non-colluding eavesdroppers is modeled by the compound wiretap channel derived in~\cite{BlochTIT2013}, where we denote the channel input at $\mathcal{S}$ as $\mc{X}$, and the channel output at the $\tau$-th receiver as $\mc{Y}_\tau$, $0\le\tau\le T$. According to \cite[Prop. 6]{BlochTIT2013}, given a certain location of the source $\bm{p}_u$, the following secrecy rate is achievable:
\begin{align}\label{eqRsec_ini}
\mathbb{R}_{\mathrm{sec}}(\bm{p}_u) = \max_{p(\mc{V},\mc{X})} \left[I(\mc{V};\mc{Y}_0) - \max_{1\le\tau\le T}I(\mc{V};\mc{Y}_\tau)\right]^+,
\end{align}
where we have $[x]^+ = \max(0,x)$, $I(x;y)$ denotes the mutual information between the random variables $x$ and $y$, $\mc{V}$ is an auxiliary random variable. The random variables $\mc{V}\rightarrow\mc{X}\rightarrow (\mc{Y}_0,\mc{Y}_\tau)$, $1\le\tau\le T$, form Markov chains, such that $\mc{V}$ and the tuples $(\mc{Y}_0,\mc{Y}_\tau)$ are statistically independent conditioned on $\mc{X}$. The outer maximization in (\ref{eqRsec_ini}) is carried out over the joint probability distribution $p(\mc{V},\mc{X})$. The secrecy rate (\ref{eqRsec_ini}) can be further optimized over the UAV's deployment location $\bm{p}_u$, which is formulated as:
\begin{align}\label{eqRsec_pb}
	\max_{\bm{p}_u} \mathbb{R}_{\mathrm{sec}}(\bm{p}_u),\quad \mathrm{s.t.}\  || \bm{p}_u - \bm{p}_u^{(1)} ||\le d_{\max}.
\end{align}

Since the instantaneous CSIs are available to the receivers, the outputs of the compound channels at the $\tau$-th receiver can be viewed as the dual-tuple $\mc{Y}_\tau = \{\bm{y}_\tau, \mb{H}_\tau\}$. By using the chain rule of mutual information, we have
\begin{align}\label{eqMI}
	I(\mc{V};\mc{Y}_\tau) &= I(\mc{V}; \bm{y}_\tau, \mb{H}_\tau)\nonumber\\
	 &= I(\mc{V};\bm{y}_\tau|\mb{H}_\tau) + I(\mc{V};\mb{H}_\tau) = I(\mc{V};\bm{y}_\tau|\mb{H}_\tau),
\end{align}
where the last equality in (\ref{eqMI}) is obtained since the instantaneous CSI $\mb{H}_\tau$ is not available to the transmitter and therefore, the mutual information becomes $I(\mc{V};\mb{H}_\tau) = 0$.

It is a challenging problem to find the distribution $p(\mc{V},\mc{X})$ in (\ref{eqRsec_ini}) maximizing the secrecy rate for compound wiretap channels. Here, we adopt the usual Gaussian signaling assumption (see, e.g. \cite{LiPetropuluTWC2011}), where we have $\mc{V} = \bm{x}$ and $\bm{x}$ is complex Gaussian distributed as $\bm{x}\sim\mc{CN}(\mb{0},\mb{Q})$ with the covariance $\mb{Q} = \mbb{E}[\bm{x}\bm{x}^\dagger]$. Therefore, the conditional mutual information $I(\mc{V};\bm{y}_\tau|\mb{H}_\tau)$ in~(\ref{eqMI}) is given by the well-known ergodic mutual information of the Gaussian MIMO channel as
\begin{align}\label{eqMI_gauss}
	I(\mc{V};\bm{y}_\tau|\mb{H}_\tau) = \mbb{E}\left\{\log\left|\mb{I}_{N_\tau} + \frac{\rho_\tau(\bm{p}_u)}{\nu}\mb{H}_\tau\mb{Q}\mb{H}_\tau^\dagger \right|\right\},
\end{align}
where the expectation is due to the definition of the conditional mutual information~\cite{Cover2006}, which is taken over the random variate $\mb{H}_\tau$. Upon applying the eigenvalue decomposition to the covariance matrix $\mb{Q}$, we obtain $\mb{Q} = P_{\max} \mb{W}\mb{\Psi}\mb{W}^\dagger$, where $0 \le \mathrm{Tr}(\mb{Q})\le P_{\max}$ and $P_{\max}$ denotes the maximum transmit power of $\mathcal{S}$, $\mb{\Psi}$ is a diagonal matrix whose $k$-th diagonal entry $\psi_k$, $1\le k\le K$, is the $k$-th largest eigenvalue of the matrix $(1/P_{\max})\mb{Q}$, and $\mb{W}$ is a unitary matrix whose $k$-th column is the eigenvector corresponding to the $k$-th eigenvalue $\psi_k$. Traditionally, $\mb{W}$ and $\mb{\Psi}$ are termed as the precoding and power allocation matrices, respectively. Note that by definition, we have the normalized power constraint $\mathrm{Tr}(\mb{\Psi}) = \sum_{k = 1}^K \psi_k \le 1$. Recalling the definition of $\rho_\tau(\bm{p}_u)$ in (\ref{eqRho}), (\ref{eqMI_gauss}) can be rewritten as
\begin{align}\label{eqCovH0}
& \mbb{E}\left\{\log\left|\mb{I}_{N_\tau} + \frac{P_{\max} c_p }{\nu ||\bm{p}_u - \bm{p}_\tau||^\alpha}\mb{H}_\tau \mb{W}\mb{\Psi}\mb{W}^\dagger \mb{H}_\tau^\dagger \right|\right\}\nonumber\\
&\qquad\qquad\qquad = \mbb{E}\left\{\log\left|\mb{I}_{N_\tau} + \frac{ \gamma }{z_\tau}\mb{H}_\tau \mb{W}\mb{\Psi}\mb{W}^\dagger \mb{H}_\tau^\dagger \right|\right\},
\end{align}
where $\gamma = P_{\max}c_p/\nu$ and $z_\tau = ||\bm{p}_u - \bm{p}_\tau||^\alpha$.

Nonetheless, it is still a challenging problem to find the optimal $\mb{W}$ and $\mb{\Psi}$ for the generic compound wiretap channels. Attempts have been made in \cite{LiTIFS2011} and \cite{LiPetropuluTWC2011} to explore the structure of the optima in some particular cases of Rician wiretap channels. Given the statistical CSI knowledge at $\mathcal{S}$, in the MISO wiretap channels, the authors of \cite{LiPetropuluTWC2011} show that the columns of the optimal precoding matrix $\mb{W}$ are the same as the eigenvectors of the covariance matrix of the legitimate channel, i.e., $\mbb{E}[\mb{H}_0^\dagger\mb{H}_0]$. Although a MISO configuration is assumed in \cite{LiPetropuluTWC2011}, the same precoder structure can be adopted for MIMO settings, which would result in beneficially harnessing the statistical channel knowledge available at $\mathcal{S}$. Specifically, upon defining $\bar{\mb{H}}_{d,0} = \sqrt{\kappa_{0}} \mb{H}_{d,0}$, the covariance matrix of $\mb{H}_0$ is calculated as
\begin{equation}
\begin{aligned}\label{eqH0H0}
\mbb{E}[\mb{H}_0^\dagger\mb{H}_0] &= \frac{1}{\kappa_{0} + 1}\bar{\mb{H}}_{d,0}^\dagger\bar{\mb{H}}_{d,0} + \frac{1}{\kappa_{0}+1} \mb{I}_{K}\\ &= \frac{1}{\kappa_{0} + 1} \mb{V}\mb{\Omega}\mb{V}^\dagger + \frac{1}{\kappa_{0}+1} \mb{I}_{K} \\ &= \frac{1}{\kappa_{0}+1}\mb{V}\left(\mb{\Omega} + \mb{I}_{K}  \right) \mb{V}^\dagger,
\end{aligned}
\end{equation}
where $\bar{\mb{H}}_{d,0}^\dagger\bar{\mb{H}}_{d,0} = \mb{V}\mb{\Omega}\mb{V}^\dagger$ represents the eigenvalue decomposition. The~($K\times K$)-element diagonal matrix obeys $\mb{\Omega} = (\mb{\Omega}^{1/2})^\dagger\mb{\Omega}^{1/2}$, where we have $\mb{\Omega}^{1/2} = \left[\mathrm{diag}(\bm{\omega})^{1/2}, \mb{0}_{N_0\times (K-N_0)}\right]$ when $K\ge N_0$, and $\mb{\Omega}^{1/2} =\left[\mathrm{diag}(\bm{\omega})^{1/2}, \mb{0}_{K\times (N_0-K)}\right]^\mathrm{T}$ when $N_0>K$. The vector $\bm{\omega} = [\omega_1,\ldots,\omega_{\min(N_0,K)}]^{\mathrm{T}}$ denotes the first $\min(N_0,K)$ number of eigenvalues of $\bar{\mb{H}}_{d,0}^\dagger\bar{\mb{H}}_{d,0}$, arranged in descending order. The columns of $\mb{V}$ are the corresponding eigenvectors. Since the scattering components $\mb{H}_{s,0}$ are uncorrelated, $\mb{V}$ also represents the eigenvectors of $\mbb{E}[\mb{H}_0^\dagger\mb{H}_0]$ as indicated in the third equality of (\ref{eqH0H0}). As $\mathcal{S}$ only has the knowledge of $\bar{\mb{H}}_{d,0}$, it is plausible for the source to transmit in the same eigenspace of $\bar{\mb{H}}_{d,0}$. Therefore, we assume using the precoder $\mb{W} = \mb{V}$ and the resultant secrecy rate is indeed achievable in this particular signal structure.\footnote{If artificial noise injection is applied at the transmitter, the transmit signal $\bm{x}$ can be reformulated as $\mathbf{x}=\mathbf{V} \mathrm{diag}(\bm{\psi}^{1/2}_{s}) \mathbf{s}+\mathbf{V} \mathrm{diag}(\bm{\psi}_{a}^{1/2}) \mathbf{a}$, where $\mathbf{s}$ and $\mathbf{a}$ denote the confidential signal and the artificial noise, while $\bm{\psi}_{s}$ and $\bm{\psi}_{a}$ represent the power allocations for the signal and the artificial noise. The joint optimization of $\bm{\psi}_{s}$ and $\bm{\psi}_{a}$ can be achieved by using a similar theoretical framework as discussed in Section III, with some straightforward derivations.}

Substituting (\ref{eqRsec_ini}) and (\ref{eqMI})-(\ref{eqCovH0}) into (\ref{eqRsec_pb}), the secrecy rate is then lower-bounded by the optimum of the following joint transmission and UAV deployment problem:
\begin{equation}
\begin{aligned}\label{eqRsec}
	R_{\mathrm{sec}} = & \max_{ \bm{\psi}\in\mbb{R}_{+}^K,\  \bm{p}_u } \left\{ R_0(\bm{\psi}, \bm{p}_u) - \max_{1\le\tau\le T} R_\tau(\bm{\psi}, \bm{p}_u) \right\}^+,\\
& \quad \mathrm{s.t.}\ \sum_{k=1}^K \psi_k\le 1,\ || \bm{p}_u - \bm{p}_u^{(1)} ||\le d_{\max},
\end{aligned}
\end{equation}
where $\mbb{R}_{+}^K$ denotes the space of the $K$-dimensional non-negative real vectors and
\begin{align}\label{eqRtau}
 R_\tau(\bm{\psi}, \bm{p}_u) &=  \mbb{E}\left\{\log\left|\mb{I}_{N_\tau} + \frac{ \gamma }{z_\tau}\mb{H}_\tau \mb{V}\mathrm{diag}(\bm{\psi})\mb{V}^\dagger \mb{H}_\tau^\dagger \right|\right\}.
\end{align}
To solve the problem (\ref{eqRsec}), the closed-form expression of the information rate $R_\tau(\bm{\psi},\bm{p}_u)$ is needed, which is equivalent to the ergodic rate of the Rician MIMO channels in conjunction with $\mb{V}\mathrm{diag}(\bm{\psi})\mb{V}^\dagger$ as the input covariance matrix. Unfortunately, such a closed-form expression is not yet available for generic Rician MIMO channels, it only exists for some special cases, such as the asymptotic expressions of the scenarios, where the number of antennas or the SNR value becomes large. In the next section, we derive an accurate approximation of the information rate of the eavesdropping channel $R_\tau(\bm{\psi},\bm{p}_u)$ for arbitrary MIMO configurations and of the input covariance matrix, which are also amenable to using standard mathematical optimization toolkits, as discussed in Section~\ref{secOpt}.

\section{Information Rate of Rician MIMO Wiretap Channels}\label{secRate}
In this section, we investigate the closed-form expressions of the information rate $R_0(\bm{\psi}, \bm{p}_u)$ of the legitimate channel, and the information rate $R_\tau(\bm{\psi},\bm{p}_u)$, $1\le\tau\le T$, of the eavesdropping channel, using different approximation techniques. Specifically, since the precoder $\mb{W}$ is aligned with the eigenspace of $\bar{\mb{H}}_{d,0}$, due to the power allocation vector $\bm{\psi}$, the original Rician MIMO channel $\mb{H}_{0}$ becomes equivalent to a lower dimensional MIMO channel having the same principal eigenvalues, where Minkowski's inequality can be applied for lower-bounding the information rate. On the other hand, the precoder $\mb{W}$ is misaligned with the eavesdropping channels hence the corresponding eigenchannels cannot be reduced. To address this issue, we apply the Haar matrix insertion~(HMI) technique to approximate the information rate $R_\tau(\bm{\psi},\bm{p}_u)$, which was previously proposed in~\cite{ZhengTCOM2019} and yields reasonably accurate agreement with the exact ergodic rate for a wide range of system settings. Finally, numerical simulations are provided for validating the proposed approximations. In what follows, we will frequently use the notations $s_\tau = \min(N_\tau, K)$, $t_\tau = \max(N_\tau, K)$, and $q_\tau = \mathrm{rank}(\mb{H}_{d,\tau})$, where the subscript $\tau$ may be ignored when it is clear from the context.

\subsection{Information Rate of Legitimate Channel}\label{secR0}
Let $\bar{\mb{H}}_{d,0} = \mb{U}\mb{\Omega}^{1/2}\mb{V}^\dagger$ denote the singular value decomposition of $\bar{\mb{H}}_{d,0}$, where $\mb{\Omega}$ is the same as in (\ref{eqH0H0}) and the columns of $\mb{U}$ and $\mb{V}$ are the left and right singular vectors of $\bar{\mb{H}}_{d,0}$ upon substituting~(\ref{eqH}) along with $\tau = 0$ into (\ref{eqRtau}), we can formulate the information rate between $\mathcal{S}$ and $\mathcal{D}$ as
\begin{align}
	&R(\bm{\psi},\bm{p}_u)  = \nonumber\\
	&\ \ \mbb{E}\left\{\log\left|\mb{I}_{N} + \frac{ \bar{\gamma} }{z}\left({\mb{\Omega}}^{1/2}  +  \bar{\mb{H}}_{s}\right) \mb{\Psi} \left( {\mb{\Omega}}^{1/2} +  \bar{\mb{H}}_{s}\right)^\dagger \right|\right\},\label{eqR0}
\end{align}
where we have dropped all the subscripts $0$ for notation simplicity, and we define the following modified parameters $\bar{\gamma} = \gamma/(\kappa_{0}+1)$ and $\bar{\mb{H}}_{s} = \mb{U}^\dagger\mb{H}_{s}\mb{V}$, where ${\bar{\mb{H}}}_{s}$ has the same probability distribution as $\mb{H}_{s}$ due to the unitary invariance of the \emph{i.i.d.} complex Gaussian matrix.

From a secrecy rate maximization perspective, it is reasonable for $\mathcal{S}$ to transmit in the eigenspace of $\mb{H}_{0}$ having positive eigenvalues, which represent the amplitudes of the eigenchannels between the $\mathcal{S}$ and $\mathcal{D}$. This scheme has the potential of increasing the information rate. On the other hand, since $\mb{\Psi}$ is a diagonal matrix, setting some of the diagonal element $\psi_k$ to zero is equivalent to deactivating the corresponding eigenchannel. Therefore, the number of activated eigenchannels is no larger than the rank of the legitimate channel $\mb{H}_0$. With the eigenvalues $\bm{\omega}$ arranged in descending order, the ($K\times K$)-element power allocation matrix $\mb{\Psi}$ is of the form $\mb{\Psi} = \mathrm{diag}([\psi_1,\ldots,\psi_r, 0, \ldots, 0]^{\mathrm{T}})$, where $\psi_i>0$, $1\le i\le r$, and $r$ denotes the number of activated eigenchannels with $r\le \min(N_0,K)$.

Let us denote $\mathrm{diag}([\psi_1,\ldots,\psi_r]^{\mathrm{T}})$ by $\widehat{\mb{\Psi}}$, and the truncation of ${\mb{\Omega}}^{1/2}$ by keeping the left ($N\times r$)-element matrix block by $\widehat{\mb{\Omega}}^{1/2}$, finally the ($N\times r$)-element \emph{i.i.d.} standard complex Gaussian distributed random matrix by $\widehat{\mb{H}}$. Then $R(\bm{\psi}, \bm{p}_u)$ in (\ref{eqR0}) can be rewritten and lower-bounded as:
\begin{align}
&R(\bm{\psi},\bm{p}_u) = \mbb{E}\left\{\log\left|\mb{I}_{r} + \frac{ \bar{\gamma} }{z} \widehat{\mb{\Psi}} \mb{G} \right|\right\} \label{eqR0_a} \\
&\ge r \mbb{E}\left\{ \log  \left[1  +  \frac{ \bar{\gamma} }{z} \exp\left( \frac{1}{r}\log\left|\widehat{\mb{\Psi}} \mb{G} \right|\right)\right]\right\}\label{eqR0_b}\\
&\ge  r \log  \left[1 + \frac{ \bar{\gamma} }{z} \exp \left( \frac{1}{r}\log \left|\widehat{\mb{\Psi}}\right| + \frac{1}{r}\mbb{E}\left\{ \log\left| \mb{G} \right| \right\} \right) \right],\label{eqR0_c}
\end{align}
where $\mb{G} = \left( \widehat{\mb{\Omega}}^{1/2} + \widehat{\mb{H}}\right)^\dagger\left(\widehat{\mb{\Omega}}^{1/2} + \widehat{\mb{H}}\right)$, (\ref{eqR0_b}) is obtained by applying Minkowski's inequality~\cite{Horn1985} and (\ref{eqR0_c}) is valid due to the convexity of the function $\log(1+a\exp(x))$ when $a>0$. Note that (\ref{eqR0_a}) indicates that the information rate of the original MIMO channel is equal to an equivalent MIMO channel having antenna dimensions of $N\times r$. Since $\widehat{\mb{\Omega}}$ preserves exactly the first $r$ eigenvalues of ${\mb{\Omega}}$, we define $\widehat{q} = \min(q,r)$ as the number of non-zero eigenvalues among $\bm{\omega}$ of the reduced channel and introduce the notation of $\widehat{\bm{\omega}} = [{\omega}_1,\ldots,{\omega}_{\widehat{q}}]^{\mathrm{T}}$. The lower bound (\ref{eqR0_c}) can be characterized explicitly by using~\cite[Thm. 7]{McKayTIT2005} as formulated in the following proposition.
\begin{proposition}[McKay and Collings~\cite{McKayTIT2005}]\label{propRL}
	The expression of the lower bound (\ref{eqR0_c}) is given by
	\begin{align}\label{eqR0_d}
		R_L(\bm{\psi},\bm{p}_u) = r \log\left(1 + \frac{ \bar{\gamma} }{z} E(\widehat{\mb{\Psi}})\right),
	\end{align}
	where
	\begin{align}
	&E(\widehat{\mb{\Psi}}) = \exp\left[ \frac{1}{r}\log\left|\widehat{\mb{\Psi}}\right| +\right.\nonumber\\
	&\quad \left. \frac{1}{r} \left( \sum_{i=1}^{r-1}\varphi(t-i) + \frac{(-1)^{\frac{\widehat{q}(\widehat{q}-1)}{2}}}{\left|\mb{V}_{\widehat{q}}(\widehat{\bm{\omega}})\right|} \sum_{j=1}^{\widehat{q}} \left|\mb{V}_{\widehat{q},j}(\widehat{\bm{\omega}})\right| \right) \right],
	\end{align}
	and $\varphi(\cdot)$ denotes the digamma function, and $\mb{V}_{\widehat{q}}(\widehat{\bm{\omega}}) = \left[{\omega}_j^{i-1}\right]$ is an $\widehat{q}\times \widehat{q}$ Vandermonde matrix. The matrix $\mb{V}_{\widehat{q},j}(\widehat{\bm{\omega}})$ is formed by replacing the $j$-th column of $\mb{V}_{\widehat{q}}(\widehat{\bm{\omega}})$ with a vector, whose $i$-th element is given by
	\begin{align}
		h_i({\omega}_j) = {\omega}_j^{i-1}\sum_{k=1}^\infty \frac{\gamma(k, {\omega}_j)}{r - \widehat{q} + i + k - 1},\quad 1\le i\le \widehat{q},
	\end{align}
	where $\gamma(k,x) = \frac{1}{\Gamma(k)}\int_0^{x}t^{k-1}e^{-t}\mathrm{d}t$ is the lower incomplete gamma function.
\end{proposition}

\subsection{Information Rate of Eavesdropping Channel}

Similar to (\ref{eqR0}), ignoring the subscript $\tau$, the information rate of the eavesdropping channel is:
\begin{align}
	&R(\bm{\psi},\bm{p}_u)\nonumber\\
	&= \mbb{E}\left\{\log\left|\mb{I}_{N} + \frac{ \bar{\gamma} }{z}\left( \bar{\mb{H}}_{d} + {\mb{H}}_{s}\right) \mb{V}\mb{\Psi}\mb{V}^\dagger \left( \bar{\mb{H}}_{d} + {\mb{H}}_{s} \right)^\dagger \right|\right\}\label{eqRtau_1}\\
	&= \mbb{E}\left\{\log\left|\mb{I}_{N} + \frac{ \bar{\gamma} }{z}\left( \widehat{\mb{H}}_{d} + \widehat{\mb{H}}_{s}\right) \widehat{\mb{\Psi}} \left( \widehat{\mb{H}}_{d} + \widehat{\mb{H}}_{s} \right)^\dagger \right|\right\},\label{eqRtau_2}
\end{align}
where $\widehat{\mb{H}}_{d}$ and $\widehat{\mb{H}}_{s}$ represent the truncated versions of $\bar{\mb{H}}_{d}\mb{V}$ and ${\mb{H}}_{s}\mb{V}$ attained by retaining their left $N\times r$ matrix blocks, respectively. However, the rate bounds of~(\ref{eqRtau_2}) obtained by following similar approaches to those in Section~\ref{secR0} or in~\cite{McKayCL2006} are inconvenient to use in the secrecy rate optimization to be presented in Section~\ref{secOpt}. This is due to the facts that: (1) When $r>N$, the Jensen-type and the Minkowski-type rate bounds require the evaluation of matrix-variate zonal polynomials~\cite{Conradie1984}, which are quite complex in numerical computations; (2) The rate bounds require the calculation of the eigenvalues of the truncated version of $\bar{\mb{H}}_{d}\mb{V}$ for a given $\widehat{\mb{\Psi}}$, which have to be re-calculated whenever the power is allocated differently. This also incurs significant amount of extra computations for solving the eigenvalue problem in the iterative optimization, since the power allocation may change from one iteration to another.

To avoid the aforementioned issues, we observe that the precoder is aligned with the eigenspace of the legitimate channel, while it is randomly projected into the eigenspace of the eavesdropping channels. Therefore, we follow the HMI technique proposed in~\cite{ZhengTCOM2019}, which approximates the information rate (\ref{eqRtau_1}) by replacing the fixed unitary matrix $\mb{V}$ by the randomly distributed unitary Haar matrix. The approximations are rendered accurate by beneficially exploiting the random relative positions between the eigenspaces of the legitimate and the eavesdropping channels. The results obtained are also relatively simple, as the eigenvalues of $\bar{\mb{H}}_{d}^\dagger\bar{\mb{H}}_{d}$ are decoupled from the power allocation $\bm{\Psi}$, which simplifies the iterative optimization.

Upon introducing the notation of $\mb{X} = \bar{\mb{H}}_{d} + {\mb{H}}_{s}$ and applying the HMI procedures proposed in \cite{ZhengTCOM2019}, an approximation $R_{U}(\bm{\psi},\bm{p}_u)$ of $R(\bm{\psi},\bm{p}_u)$ in (\ref{eqRtau_1}) can be constructed as
\begin{align}
  R_{U}(\bm{\psi},\bm{p}_u) = \log\mbb{E}\left\{\left|\mb{I}_K + \frac{\bar{\gamma}}{z} \mb{X}^\dagger\mb{X}\mb{T}\mb{\Psi}\mb{T}^\dagger\right|\right\},\label{eqRtau_appx}
\end{align}
where $\mb{T}\in\mc{U}_K$ is a random unitary Haar matrix and $\mc{U}_K$ is the unitary group containing all $K\times K$ unitary matrices \cite{Sternberg1995}. Note that the expectation in (\ref{eqRtau_appx}) is taken over both the random matrix $\mb{X}$ and the random Haar matrix $\mb{T}$.

Without loss of generality, we assume that the elements of the vector $\bm{\psi}$ are ordered and the first $r$ elements are non-zero, i.e., we have $\psi_1\ge \psi_2\ge\ldots\ge \psi_r>0$ and $\psi_{r+1}=\ldots=\psi_K = 0$. Recalling that $\bar{\mb{H}}_{d}^\dagger\bar{\mb{H}}_{d}$ has $q$ non-zero eigenvalues, so that $\omega_1\ge\ldots\ge\omega_{q}>0$ and $\omega_{q+1} = \ldots = \omega_K = 0$, we define the following notations for $1\le i\le (s+1)$ and give the expressions of $R_{U}(\bm{\psi},\bm{p}_u)$ in Proposition~\ref{propRtau}:
\begin{align}
a_{i,j} &= \Gamma(t-s+i+j-1),\quad  1\le j\le s-q,\\
b_{i,j} &= e^{\omega_j}\sum_{k=0}^{i-1} \frac{\Gamma(i)\Gamma(t-s+i)\Gamma(t-s+1)\omega_j^k}{\Gamma(i-k)\Gamma(k+1)\Gamma(t-s+k+1)},\nonumber\\&\qquad\qquad\qquad\qquad\qquad\qquad\qquad\qquad  1\le j\le q.
\end{align}
The closed-form expressions of $R_U(\bm{\psi},\bm{p}_u)$ under different configurations of $K$ and $N$ are summarized in the following proposition.

\begin{proposition}\label{propRtau}
The expressions of $R_U(\bm{\psi},\bm{p}_u)$ are given as follows:
\begin{align}\label{eqRtau_4}
  R_U(\bm{\psi},\bm{p}_u) = g_0 + g_1(\bm{\psi},z),
\end{align}
where
\begin{align*}
	& g_0 = \sum_{j = 1}^{\min(r,s)}\log\frac{\Gamma(K+1-j)\Gamma(j+1)}{\Gamma(K+1)}\nonumber\\
	&\qquad\qquad\qquad - \log\left|
	\begin{array}{ll}
		\left\{ a_{i,j} \right\}_{s\times (s-q)} &
		\left\{ b_{i,j} \right\}_{s\times q}
	\end{array}
	\right|,\nonumber\\
	& g_1(\bm{\psi}, z) = \log\frac{\left|\mb{D}(\bm{\psi},z)\right|}{\left|\mb{V}_r(\bm{\psi})\right|} +  \frac{r(r-1)}{2}\log\left(\frac{z}{\bar{\gamma}}\right).
\end{align*}
The definition of Vandermonde matrix $\mb{V}_r(\bm{\psi})$ follows the one in Proposition~\ref{propRL}. The matrix $\mb{D}(\bm{\psi},z)$ is given in (\ref{eqD1}) and (\ref{eqD2}) on top of the next page.

\begin{figure*}[tb]
	\begin{align}
		&\mbox{When }s\ge r: \mb{D}(\bm{\psi},z) = \left[
		\begin{array}{ll}
			\left\{ a_{j,i} \right\}_{(s-q)\times (s-r)} & \left\{\sum\limits_{n=s-r}^{s} {K\choose s-n} a_{n+1,i}  \left(\frac{\bar{\gamma}}{z}\psi_j\right)^{n+r-s}  \right\}_{(s-q)\times r}
			\\
			\left\{ b_{j,i} \right\}_{q\times (s-r)} & \left\{ \sum\limits_{n=s-r}^s {K\choose s-n} b_{n+1,i} \left(\frac{\bar{\gamma}}{z}\psi_j\right)^{n+r-s} \right\}_{q\times r}
		\end{array}
		\right].\label{eqD1}\\
		&\mbox{When }s< r:\mb{D}(\bm{\psi},z) = \left[
		\begin{array}{l}
			\left\{ \left(\frac{\bar{\gamma}}{z}\right)^{i-1}\psi_j^{i-1} \right\}_{(r-s)\times r}\\
			\left\{ \sum\limits_{n=0}^s {K\choose s-n} a_{n+1,i}\left(\frac{\bar{\gamma}}{z}\psi_j\right)^{n+r-s} \right\}_{(s-q)\times r} \\
			\left\{ \sum\limits_{n=0}^s {K\choose s-n} b_{n+1,i}\left(\frac{\bar{\gamma}}{z}\psi_j\right)^{n+r-s} \right\}_{q\times r}
		\end{array}
		\right].\label{eqD2}
	\end{align}
	\hrulefill
\end{figure*}
\end{proposition}
\begin{proof}
  The proof is provided in Appendix~\ref{appxRtau}.
\end{proof}

\subsection{Numerical Validations of Rate Approximations}

In this section, we present numerical justifications for the approximate information rates of both the legitimate and of the eavesdropping channels, respectively. Note that the numerical results in this section are the information rates under arbitrary power allocation $\bm{\psi}$ and under $\bm{p}_u$ without optimization. In Figs.~\ref{figR_error} (a) and (b), we plot the Empirical Cumulative Distribution Functions (ECDFs) of the relative approximation errors $(R_L(\bm{\psi}, \bm{p}_u) - R_0(\bm{\psi},\bm{p}_u))/R_0(\bm{\psi},\bm{p}_u)$ between $R_L$ given in Proposition~\ref{propRL} and the simulated $R_0$. We randomly generate $10^6$ samples of the power allocation vector $\bm{\psi}$ in each simulated case, and $\bm{p}_u$ is selected for showing that the received SNR is 10~dB in Fig.~\ref{figR_error} (a) and 20~dB in Fig.~\ref{figR_error} (b), respectively. For comparison, we also plot the ECDFs of the relative errors when $R_0$ is approximated by $R_U$ in Proposition~\ref{propRtau}. It is clear that for all the MIMO configurations tested, the relative approximation errors are sufficiently small upon using $R_L$, which are on the order of $10^{-2}$ when $K = N = 4$, $10^{-3}$ when $K = 8$ and $N = 4$, and $10^{-4}$ when $K = 4$ and $N = 8$. On the other hand, the approximation errors are significantly higher upon using $R_U$ and may over-estimate $R_0$, which is not suitable for our secrecy rate optimization, since the reliability requirements may become violated.

In Figs.~\ref{figR_error} (c) and (d), under the same simulation settings, we plot the ECDFs of the relative approximation errors of the proposed approximation $R_{U}$ given in Proposition~\ref{propRtau} for the information rate $R_\tau$ when the received SNRs are 10 dB and 20 dB, respectively. Additionally, we compare the proposed approximations to the upper bound in~\cite{McKayTIT2005} by using Jensen's inequality, which, to the best of our knowledge, is the most suitable one for the secrecy rate optimization considered in the open literature, since it actually over-estimates the information rate of the eavesdropping channel and hence it dose not violate the secrecy constraint. We also note that other analytical results of the information rate of Rician MIMO schemes do exist, e.g., see~\cite{McKayCL2006, MatthaiouTCOM2009}. However, these results are either excessively complex for numerical evaluation, or under-estimate the information rate of the eavesdropping channel, which may violate the secrecy constraint. The results show that the proposed approximation (\ref{eqRtau_4}) adequately estimates the exact information rate of the eavesdropping channel at a relative approximation error on the order of $10^{-3}\sim10^{-2}$ for the scenarios considered. Therefore, the numerical results of Fig.~\ref{figR_error} validate that the $R_L$ given in Proposition~\ref{propRL} and $R_U$ given in Proposition~\ref{propRtau} are consistent with our simulations achieving sufficiently accurate agreement. Additionally, the closed-form expressions obtained are explicit in terms of the optimization variables $\bm{\psi}$ and $\bm{p}_u$, which can be therefore applied in our secrecy rate optimization problem.

\begin{figure*}[t!]
  \centering
  \subfigure[]{\includegraphics[width=2.6in]{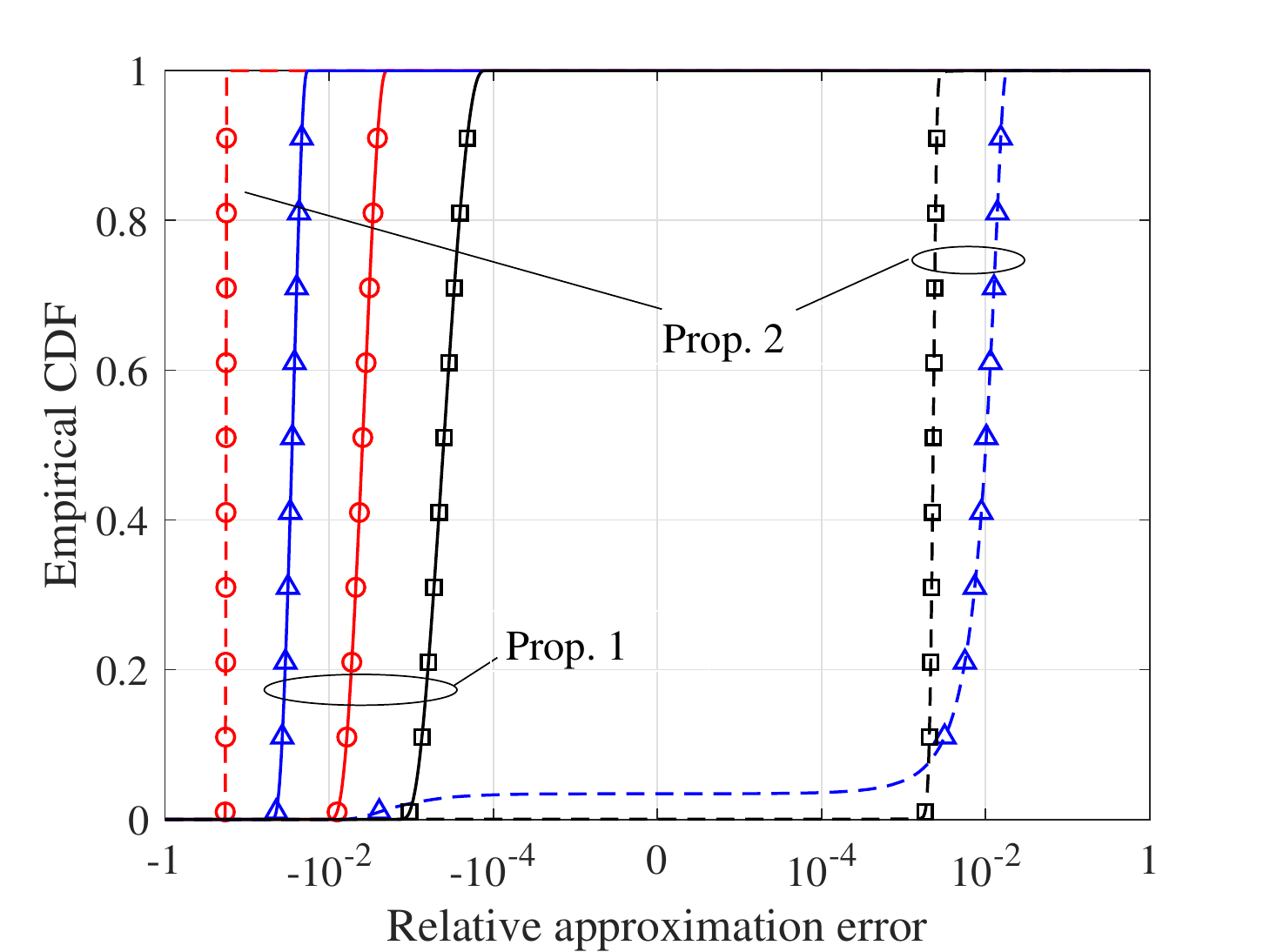}}
  \subfigure[]{\includegraphics[width=2.6in]{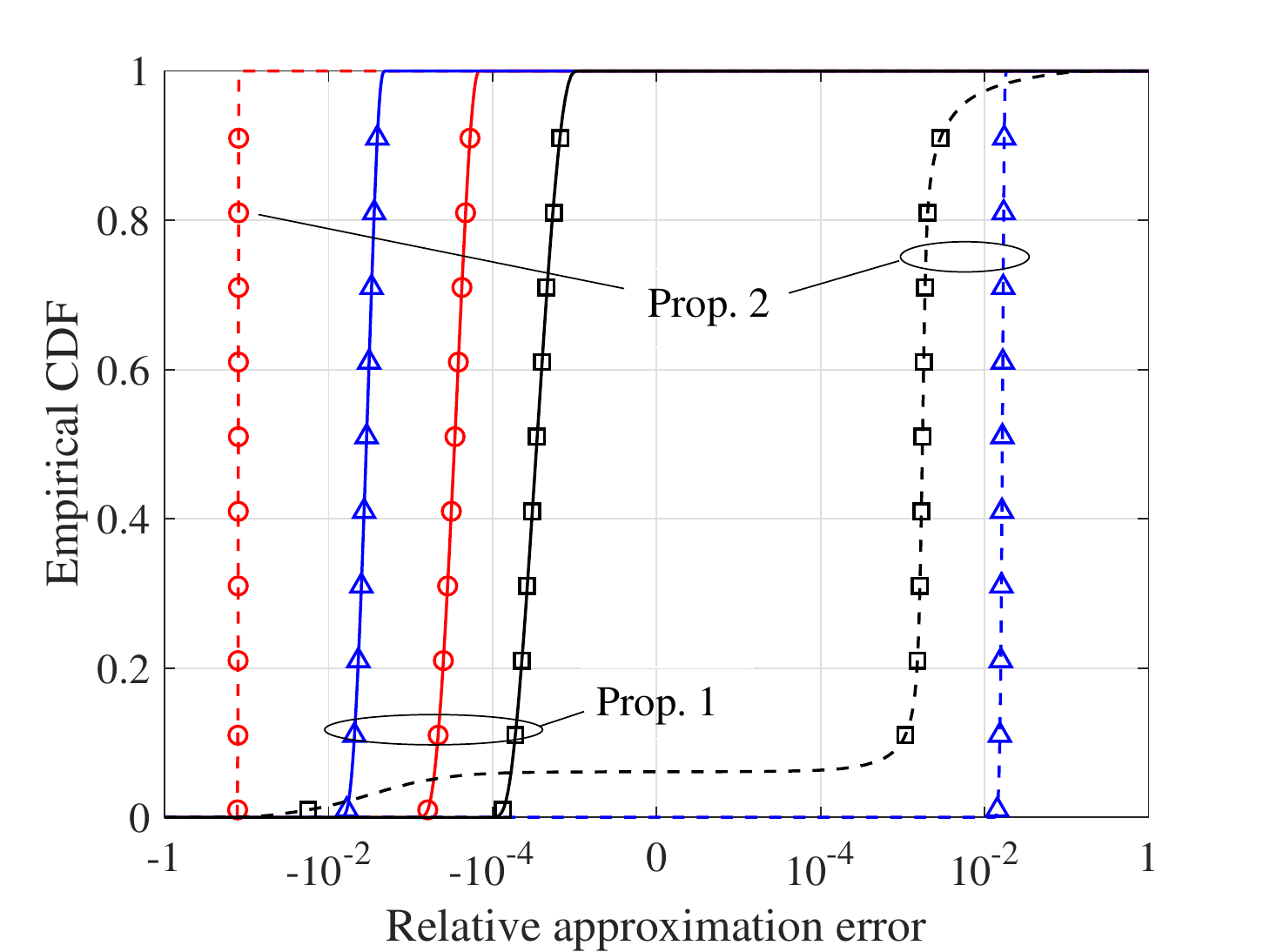}}
  \subfigure[]{\includegraphics[width=2.6in]{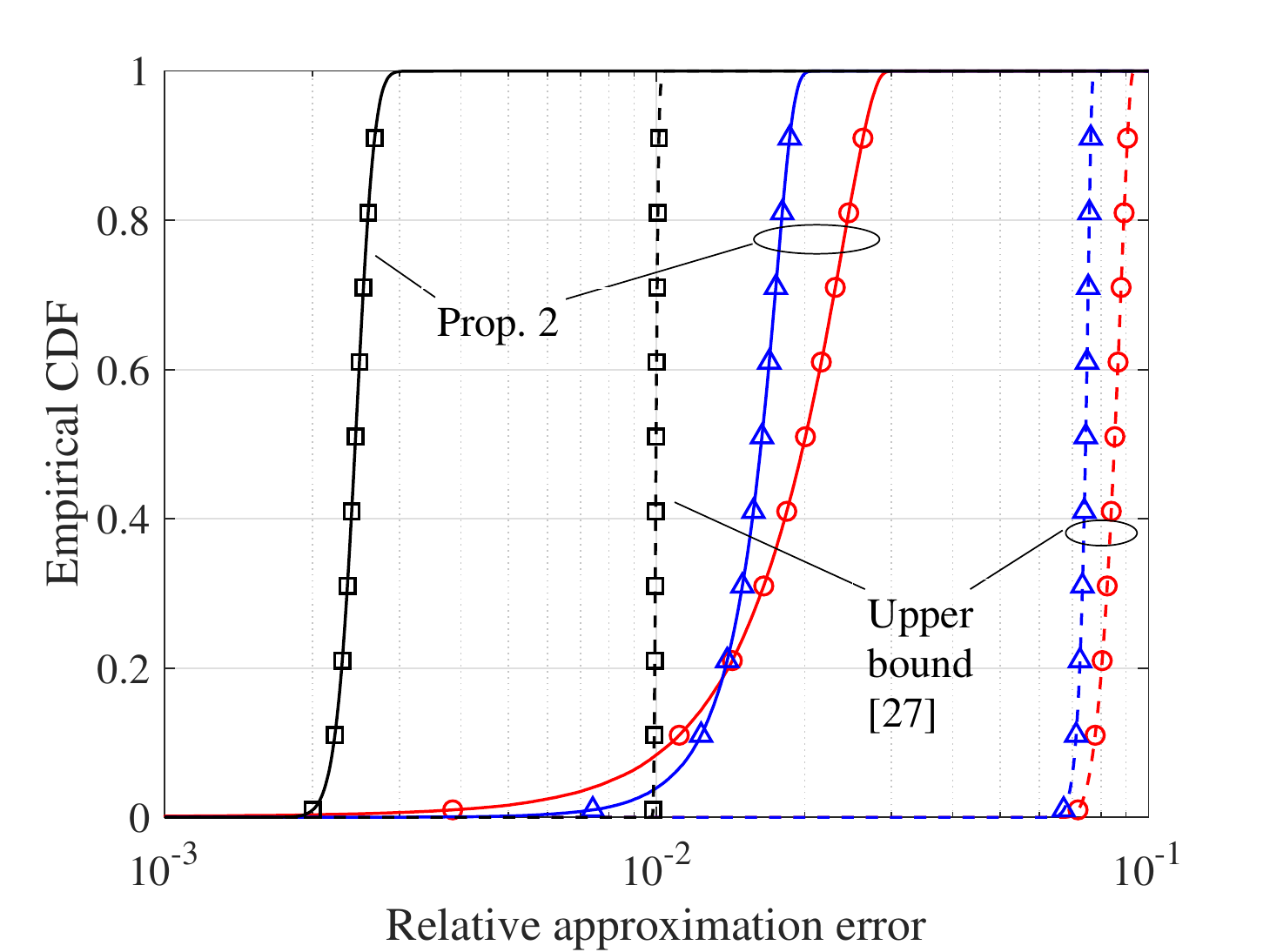}}
  \subfigure[]{\includegraphics[width=2.6in]{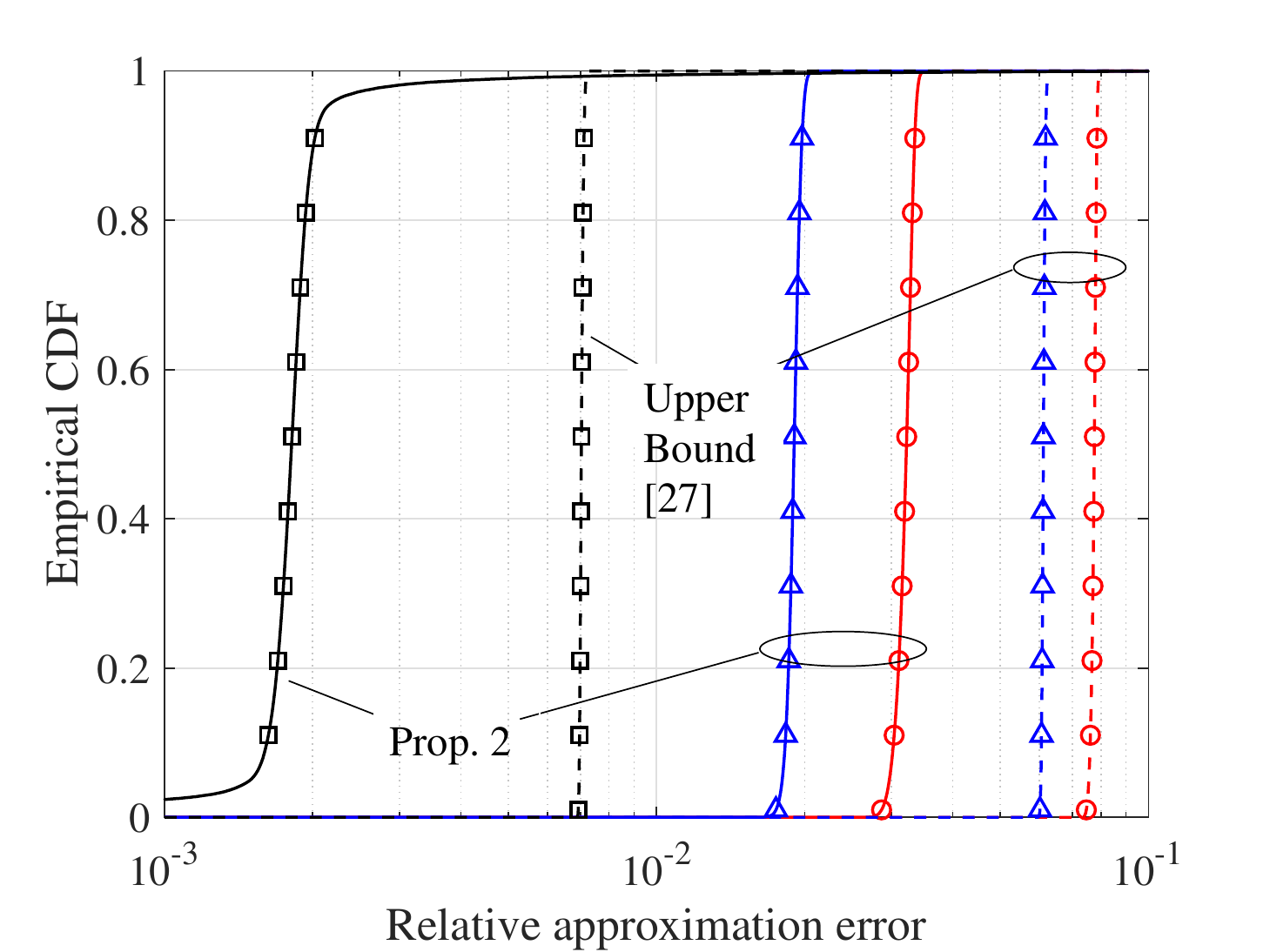}}
  \caption{Cumulative distribution function of relative errors between the approximate information rates and the simulated rates. \lq$\circ$\rq: $K = 8$ and $N = 4$; \lq$\triangle$\rq: $K = N = 4$; \lq$\square$\rq: $K = 4$ and $N = 8$. (a) and (b): ECDFs of approximation errors of $R_0$ at SNRs 10~dB and 20~dB; (c) and (d): ECDFs of approximation errors of $R_1$ at SNRs 10~dB and 20~dB.}
  \label{figR_error}
\end{figure*}

\section{Secrecy Rate Optimizations}\label{secOpt}
In this section, we conceive an alternating algorithm for the joint UAV transceiver and location optimization, where the power allocation vector $\bm{\psi}$ and the hovering location $\bm{p}_u$ are optimized iteratively. Since the objective function of the problem (\ref{eqRsec}) does not have closed-form expressions, we adopt the information rate derived in Section~\ref{secRate} for the legitimate and the eavesdropping channels to construct the approximation of the objective function. Specifically, $R_0$ and $R_\tau$, $1\le \tau \le T$, in (\ref{eqRsec}) are replaced by $R_L$ of (\ref{eqR0_d}) and $R_{U,\tau}$ of (\ref{eqRtau_4}), respectively, yielding the following problem:
    \begin{equation}
	\begin{aligned}\label{eqRsec_appx}
	& \max_{ \bm{\psi}\in\mbb{R}_{+}^K,\  \bm{p}_u } \left\{ R_L(\bm{\psi}, \bm{p}_u) - \max_{1\le\tau\le T} R_{U,\tau}(\bm{\psi}, \bm{p}_u) \right\}^+, \\
& \quad\quad \mathrm{s.t.}\ \sum_{k=1}^K \psi_k\le 1,\ || \bm{p}_u - \bm{p}_u^{(1)} ||\le d_{\max}.
	\end{aligned}
    \end{equation}
Note that the objective function of (\ref{eqRsec_appx}) is an estimate of the objective function's lower bound in (\ref{eqRsec}) for any arbitrary power allocation $\bm{\psi}$ and location $\bm{p}_u$. Hence, the optimum of (\ref{eqRsec_appx}) is an estimate of the secrecy rate's lower bound.

The optimization of (\ref{eqRsec_appx}) over $\bm{\psi}$ and $\bm{p}_u$ is still difficult due to the following pair of reasons: (1) The optimization variables $\bm{\psi}$ and $\bm{p}_u$ are coupled, since the optimal transmit signal is determined by the spatial structure of the wireless channel, especially when the location-dependent LoS component dominates the end-to-end propagation channel; (2) For each variable $\bm{\psi}$ and $\bm{p}_u$, the problem considered in (\ref{eqRsec_appx}) is non-convex. That is, for $\bm{\psi}$ the objective function of (\ref{eqRsec_appx}) is given by the difference of two concave functions, while for $\bm{p}_u$ the convexity of the objective function is indetermined. To address these problems, we adopt the alternating optimization framework that iteratively finds the local optima of $\bm{\psi}$ and $\bm{p}_u$ by fixing the other. In the $\bm{\psi}$-optimizing subproblem, $R_{U,\tau}(\bm{\psi},\bm{p}_u)$ of (\ref{eqRsec_appx}) is expanded via the first-order Taylor expansion around a given $\bm{\psi}$, thus converting the objective function into a concave function of $\bm{\psi}$. In the $\bm{p}_u$-optimizing subproblem, both $R_L(\bm{\psi},\bm{p}_u)$ and $R_{U,\tau}(\bm{\psi},\bm{p}_u)$ are expanded via the Taylor expansion, thus converting the objective function into an affine function with respect to $\bm{p}_u$. Both subproblems can be solved by standard optimization toolkits, such as CVX~\cite{grant2013}. By updating the expanded points $\bm{\psi}$ and $\bm{p}_u$ of the Taylor series, the proposed algorithm is expected to converge to a local optimum of (\ref{eqRsec_appx}).

\textbf{Step 1: Update $\bm{\psi}$.} Given the $n$-th updates of $\bm{\psi}^{(n)}$ and $\bm{p}_u^{(n)}$, the evolution from $\bm{\psi}^{(n)}$ to $\bm{\psi}^{(n+1)}$ can be converted into a series of convex optimization subproblems formulated as:
\begin{align} \label{eqRsec_psi_sca}
	\bm{\phi}^{(m+1)} = &\argmax_{ \bm{\phi}\in\mbb{R}_{+}^K }\left\{ \mc{R}_1\left(\bm{\phi}, \bm{p}_u^{(n)}\middle| \bm{\phi}^{(m)}\right)   \right\}, \nonumber \\
	&\textrm{s.t.}\ \sum_{i=1}^K \phi_{i}\le 1,
\end{align}
where $\mc{R}_1\left(\bm{\phi},\bm{p}_u^{(n)}\middle|\bm{\phi}^{(m)}\right)\!=\! R_L\left(\bm{\phi}, \bm{p}_u^{(n)}\right)\! -\! \max\limits_{1\le\tau\le T}\! R_{U,\tau}\left(\bm{\phi}, \bm{p}_u^{(n)}\middle| \bm{\phi}^{(m)} \right)$, $\bm{\phi}^{(1)} = \bm{\psi}^{(n)}$ and $\bm{\phi}^{(m^*)} = \bm{\psi}^{(n+1)}$ are the initial and the optimized values of the $n$-th update of the power allocation vector, and $m^{*}$ denotes the number of iterations of $\bm{\phi}$ needed in the $n$-th update. The function $R_{U,\tau}\left(\bm{\phi}, \bm{p}_u^{(n)}\middle| \bm{\phi}^{(m)} \right)$ is the first-order Taylor expansion of $R_{U,\tau}\left(\bm{\phi}, \bm{p}_u^{(n)} \right)$ at the point $\bm{\phi}^{(m)}$, which is expressed as:
\begin{align}
R_{U,\tau} & \left(\bm{\phi}, \bm{p}_u^{(n)}\middle| \bm{\phi}^{(m)} \right) = R_{U,\tau}\left(\bm{\phi}^{(m)}, \bm{p}_u^{(n)} \right) \nonumber \\
& + \left[ \nabla_{\bm{\phi}} R_{U,\tau}\left(\bm{\phi}, \bm{p}_u^{(n)} \right) \middle| {}_{\bm{\phi} = \bm{\phi}^{(m)}} \right]^\mathrm{T}\left[\bm{\phi} - \bm{\phi}^{(m)}\right],\label{eqRtau_taylor}
\end{align}
where $\nabla_{\bm{\phi}} R_{U,\tau}\left(\bm{\phi}, \bm{p}_u^{(n)} \right) = \left[\frac{\partial}{\partial\phi_{1}} R_{U,\tau}\left(\bm{\phi}, \bm{p}_u^{(n)} \right), \ldots, \frac{\partial}{\partial\phi_{K}} R_{U,\tau}\left(\bm{\phi}, \bm{p}_u^{(n)} \right)\right]^{\mathrm{T}}$ denotes the gradient with respect to $\bm{\phi}$. After applying Jacobi's formula \cite{Magnus1999}, the derivative $\frac{\partial}{\partial \phi_{j}} R_{U,\tau}\left(\bm{\phi}, \bm{p}_u^{(n)} \right)$ is obtained as $\frac{\partial}{\partial \phi_{j}} R_{U,\tau}\left(\bm{\phi}, \bm{p}_u^{(n)} \right) = \mathrm{Tr}\left[\mb{D}\left(\bm{\phi},z_\tau^{(n)}\right)^{-1}\frac{\partial}{\partial \phi_{j}}\mb{D}\left(\bm{\phi},z_\tau^{(n)}\right)\right] - \mathrm{Tr}\left[\mb{V}_r(\bm{\phi})^{-1}\frac{\partial}{\partial \phi_{j}}\mb{V}_r(\bm{\phi})\right]$, where $z_\tau^{(n)} = ||\bm{p}_u^{(n)} - \bm{p}_\tau||^\alpha$. The matrix derivative $\frac{\partial}{\partial \phi_{j}}\mb{V}_r(\bm{\phi})$ is calculated by taking the derivative of each of the matrix elements with respect to $\phi_{j}$ and it is therefore given by $\left[\frac{\partial}{\partial \phi_{j}}\mb{V}_r(\bm{\phi})\right]_{k,l} = (k-1)\phi_{l}^{k-2}$ when $l = j$ and $1\le k\le r$, and equals $0$ otherwise.

Similarly, when $s\ge r$ and $r>s$, the matrix derivatives $\frac{\partial}{\partial \phi_{j}}\mb{D}\left(\bm{\phi},z_\tau^{(n)}\right)$ are calculated by (\ref{S_derivation}) and (\ref{r_derivation}), respectively, and are shown on top of the next page.
\begin{figure*}[tb]
\begin{align} \label{S_derivation}
\left[\frac{\partial}{\partial \phi_{j}}\mb{D}\left(\bm{\phi},z_\tau^{(n)}\right)\right]_{k,l} &=
\begin{cases}
\sum\limits_{i=1}^r {K\choose r-i} \frac{i\, a_{i+s-r+1,k}}{\phi_{l+r-s}^{1-i}}\left(\frac{\bar{\gamma}}{z_\tau^{(n)}}\right)^{i}, & l = s-r+j, 1\le k\le s-q;\\
\sum\limits_{i=1}^r {K\choose r-i} \frac{i\, b_{i+s-r+1,k-s+q}}{\phi_{l+r-s}^{1-i}}\left(\frac{\bar{\gamma}}{z_\tau^{(n)}}\right)^{i}, & l = s-r+j, s-q+1\le k\le s;\\
0, & \mathrm{otherwise}.
\end{cases}
\end{align}
\begin{align} \label{r_derivation}
\left[\frac{\partial}{\partial \phi_{j}}\mb{D}\left(\bm{\phi}, z_\tau^{(n)}\right)\right]_{k,l} &= \begin{cases}
(k-1)\left(\frac{\bar{\gamma}}{z_\tau^{(n)}}\right)^{k-1}\phi_{l}^{k-2}, & l=j, 1\le k\le r-s;\\
\sum\limits_{i=r-s}^r {K\choose r-i} \frac{i\, a_{i+s-r+1,k+s-r}}{\phi_{l}^{1-i}}  \left(\frac{\bar{\gamma}}{z_\tau^{(n)}}\right)^{i}, & l = j, r-s+1\le k\le r-q;\\
\sum\limits_{i=r-s}^r {K\choose r-i} \frac{i\, b_{i+s-r+1,k+q-r}}{\phi_{l}^{1-i}} \left(\frac{\bar{\gamma}}{z_\tau^{(n)}}\right)^{i}, & l = j, r-q+1\le k\le r;\\
0, & \mathrm{otherwise}.
\end{cases}
\end{align}
\hrulefill
\end{figure*}
Since the first-order Taylor expansion in (\ref{eqRtau_taylor}) is affine with respect to the variable $\bm{\phi}$, the subproblem (\ref{eqRsec_psi_sca}) is concave and can be solved by the CVX package.

\textbf{Step 2: Update $\bm{p}_u$.} Following similar procedures as in Step 1, the update of $\bm{p}_u$ can also be determined by optimizing the Taylor-expanded version of (\ref{eqRsec_appx}). Since both $R_L$ and $R_{U,\tau}$ are non-convex with respect to $\bm{p}_u$, given $\bm{\psi}^{(n+1)}$ and $\bm{p}_u^{(n)}$, the $n$-th update $\bm{p}_u^{(n+1)}$ is obtained by
	\begin{align}\label{eqRsec_pb_sca}
	\bm{p}_u^{(n+1)} & = \argmax_{ \bm{p}_u }\left\{ \mc{R}_2\left(\bm{\psi}^{(n+1)}, \bm{p}_u  \middle| \bm{p}_u^{(n)} \right) \right\}, \nonumber \\
& \mathrm{s.t.}\ || \bm{p}_u - \bm{p}_u^{(1)} ||\le d_{\max},\ ||\bm{p}_u - \bm{p}_u^{(n)}||\le d_\Delta,
	\end{align}
where $\mc{R}_2\left(\bm{\psi}^{(n+1)}, \bm{p}_u  \middle| \bm{p}_u^{(n)} \right) = R_L\left(\bm{\psi}^{(n+1)},\bm{p}_u \middle| \bm{p}_u^{(n)}\right) - \max_{1\le\tau\le T} R_{U,\tau}\left(\bm{\psi}^{(n+1)},\bm{p}_u \middle| \bm{p}_u^{(n)}\right)$. Here, the functions $R_L\left(\bm{\psi}^{(n+1)},\bm{p}_u \middle| \bm{p}_u^{(n)}\right)$ and $R_{U,\tau}\left(\bm{\psi}^{(n+1)},\bm{p}_u\middle| \bm{p}_u^{(n)} \right)$ are given by
\begin{align}
	& R_X\left(\bm{\psi}^{(n+1)},\bm{p}_u \middle| \bm{p}_u^{(n)}\right) = R_X\left(\bm{\psi}^{(n+1)},\bm{p}_u^{(n)} \right) \nonumber \\
& + \left[\nabla_{\bm{p}_u} R_X\left(\bm{\psi}^{(n+1)},\bm{p}_u \right)\middle|_{\bm{p}_u = \bm{p}_u^{(n)}}\right]^{\mathrm{T}}\left(\bm{p}_u - \bm{p}_u^{(n)}\right),
\end{align}
where $X\in\{L,U\}$ and $\nabla_{\bm{p}_u} R_X\left(\bm{\psi}^{(n+1)},\bm{p}_u \right) = \frac{\partial}{\partial z_{\tau}} R_X(\bm{\psi}^{(n+1)},\bm{p}_u)\cdot\frac{\partial z_\tau}{\partial \bm{p}_u} = 2\left[\frac{\partial}{\partial z_{\tau}} R_X(\bm{\psi}^{(n+1)},\bm{p}_u)\right](\bm{p}_u - \bm{p}_\tau)$. Recalling $R_L$ given in (\ref{eqR0_d}), its derivative with respect to $z_0$ is calculated as $\frac{\partial}{\partial z_{0}} R_L\left(\bm{\psi}^{(n+1)},\bm{p}_u\right) = - r \frac{\bar{\gamma} E\left(\widehat{\mb{\Psi}}^{(n+1)}\right) }{z_0^2 + z_0\bar{\gamma} E\left(\widehat{\mb{\Psi}}^{(n+1)}\right) }$, where $\widehat{\mb{\Psi}}^{(n+1)} = \mathrm{diag}\left(\left[\psi_1^{(n+1)},\ldots,\psi_r^{(n+1)}\right]^{\mathrm{T}}\right)$. The derivative $\frac{\partial}{\partial z_\tau}R_{U,\tau}\left(\bm{\psi}^{(n+1)},\bm{p}_u\right)$ is given by
\begin{align}
	& \frac{\partial}{\partial z_\tau}R_{U,\tau}\left(\bm{\psi}^{(n+1)},\bm{p}_u\right) = \frac{r(r-1)}{2 z_\tau} + \nonumber \\
&\qquad \mathrm{Tr}\left[\mb{D}\left(\bm{\psi}^{(n+1)}, z_\tau\right)^{-1}\frac{\partial}{\partial z_\tau}\mb{D}\left(\bm{\psi}^{(n+1)},z_\tau\right)\right],
\end{align}
where the derivatives $\frac{\partial}{\partial z_\tau}\mb{D}\left(\bm{\psi}^{(n+1)},z_\tau\right)$ are calculated by (\ref{S_derivation_z}) and (\ref{r_derivation_z}), when $r\le s$ and $r>s$, respectively, and are shown on top of the next page.
\begin{figure*}[tb]
\begin{align} \label{S_derivation_z}
\left[\frac{\partial}{\partial z_\tau}\mb{D}\left(\bm{\psi}^{(n+1)},z_\tau\right)\right]_{k,l} &=
\begin{cases}
\sum\limits_{i=1}^r {K\choose r-i} \frac{i\, a_{i+s-r+1,k}}{-z_\tau^{i+1}(\bar{\gamma}\psi_{l+r-s}^{(n+1)})^{-i}}, & s-r+1\le l\le s, 1\le k\le s-q;\\
\sum\limits_{i=1}^r {K\choose r-i} \frac{i\, b_{i+s-r+1,k-s+q}}{-z_\tau^{i+1}(\bar{\gamma}\psi_{l+r-s}^{(n+1)})^{-i}}, & s-r+1\le l\le s, s-q+1\le k\le s;\\
0, & \mathrm{otherwise}.
\end{cases}
\end{align}
\begin{align} \label{r_derivation_z}
	\left[\frac{\partial}{\partial z_\tau}\mb{D}\left(\bm{\psi}^{(n+1)},z_\tau\right)\right]_{k,l} &= \begin{cases}
	-(k-1)\frac{\left(\bar{\gamma}\psi_{l}^{(n+1)}\right)^{k-1}}{z_{\tau}^{k}}, & 1\le l\le r, 1\le k\le r-s;\\
	\sum\limits_{i=r-s}^r {K\choose r-i} \frac{i\, a_{i+s-r+1,k+s-r}}{-z_{\tau}^{i+1}\left(\bar{\gamma} \psi_{l}^{(n+1)}\right)^{-i}}, & 1\le l\le r, r-s+1\le k\le r-q;\\
	\sum\limits_{i=r-s}^r {K\choose r-i} \frac{i\, b_{i+s-r+1,k+q-r}}{-z_{\tau}^{i+1}\left(\bar{\gamma} \psi_{l}^{(n+1)}\right)^{-i}}, & 1\le l\le r, r-q+1\le k\le r;\\
	0, & \mathrm{otherwise}.
	\end{cases}
\end{align}
\hrulefill
\end{figure*}

Note that an additional constraint $||\bm{p}_u - \bm{p}_u^{(n)}||\le d_\Delta$ is added in (\ref{eqRsec_pb_sca}) to avoid overly large displacement of the UAV between two consecutive updates, which may cause inaccurate approximation of the Taylor expansion. Since the objective function of (\ref{eqRsec_pb_sca}) is an affine function with respect to $\bm{p}_u$, the problem (\ref{eqRsec_pb_sca}) can also be solved by standard optimization packages.

Finally, the optimizations of $\bm{\psi}$ and $\bm{p}_u$ are alternated until convergence is achieved. The iterative power optimization and the complete joint optimization procedures of the power and the location are summarized in Algorithms~\ref{alg_power} and~\ref{alg_alter}, respectively. The initial power allocation $\bm{\psi}^{(1)} = \bm{\psi}_{\mathrm{WF}}$ is selected as the water-filling solution over $\mb{H}_{d,0}$ of the main LoS channel, and the location variable $\bm{p}_u^{(1)} = \bm{p}_u^{\mathrm{start}}$ is the initial dispatch location of the UAV. Both Algorithms~\ref{alg_power} and~\ref{alg_alter} are terminated when the changes of the objective functions between the consecutive steps are lower than thresholds $\epsilon_1$ and $\epsilon_2$, respectively.

Since both of Algorithms 1 and 2 rely on the CVX toolbox, which solves convex problems via the interior point method (IPM), the computation complexity of these algorithms can be explicitly analyzed by leveraging the analytical treatment in~[52]. Specifically, the power allocation problem (28) in line 3 of Algorithm 1 can be approximated as a conic quadratic programming (CQP) problem, where the complexity scales as $\mathcal{O}\left(T \sqrt{m+1}\left(n^{3}+n^{2} m+n m\right)\right)$ with $n = m = K$. The location optimization problem (32) in line 4 of Algorithm 2 has an affine objective function and conic quadratic constraint, hence it is also a CQP, where the computational complexity scales as $\mathcal{O}\left(T(m+n)^{3 / 2} n^{2}\right)$, with $m = 1$ and $n = 3$. Denoting the numbers of iterations in Algorithms 1 and 2 as $L_1$ and $L_2$, respectively, while retaining only the highest order term, the overall complexity of Algorithm 2 scales as $\mc{O}(L_1 L_2 T K^{3.5})$, which is of polynomial order.

\begin{algorithm}[t]
	\caption{Iterative power optimization}\label{alg_power}
	\begin{algorithmic}[1]
		\State \textbf{set} $m = 1$, $\bm{\phi}^{(1)} = \bm{\psi}^{(n)}$, $\epsilon_1>0$
		\State \textbf{repeat}
		\State\qquad Solve the subproblem (\ref{eqRsec_psi_sca}) via CVX and set the output as $\bm{\phi}^{(m+1)}$;
		\State\qquad $m = m+1$;
		\State \textbf{until} $\mc{R}_1\left(\bm{\phi}^{(m+1)}, \bm{p}_u^{(n)}\middle| \bm{\phi}^{(m)}\right) - \mc{R}_1\left(\bm{\phi}^{(m)}, \bm{p}_u^{(n)}\middle| \bm{\phi}^{(m-1)}\right)<\epsilon_1$;
		\State \textbf{set} $\bm{\psi}^{(n+1)} = \bm{\phi}^{(m)}$.
	\end{algorithmic}
\end{algorithm}

\begin{algorithm}[t]
	\caption{Alternating power and deployment optimization}\label{alg_alter}
	\begin{algorithmic}[1]
		\State \textbf{set} $n = 1$, $\bm{\psi}^{(1)} = \bm{\psi}_{\mathrm{WF}}$, $\bm{p}_u^{(1)} = \bm{p}_u^{\mathrm{start}}$, $\epsilon_2>0$
		\State \textbf{repeat}
		\State\qquad Update $\bm{\psi}^{(n+1)}$ as the output of Algorithm~\ref{alg_power};
		\State\qquad Update $\bm{p}_u^{(n+1)}$ by solving the subproblem (\ref{eqRsec_pb_sca}) via CVX;
		\State\qquad $n = n + 1$;
		\State \textbf{until} $\mc{R}_2\left(\bm{\psi}^{(n)}, \bm{p}_u^{(n)}  \middle| \bm{p}_u^{(n)} \right) - \mc{R}_2\left(\bm{\psi}^{(n-1)}, \bm{p}_u^{(n-1)}  \middle| \bm{p}_u^{(n-1)} \right) \le \epsilon_2$.
	\end{algorithmic}
\end{algorithm}

\section{Numerical Results}\label{secResult}

In this section, numerical simulations are conducted for quantifying the impact of the number of antennas, the locations of the source UAV, the destination, and the eavesdroppers on the achievable secrecy rate via the joint optimization of the transceiver and the UAV location. In the following, the eavesdroppers are placed closer to the initial location of the UAV, than the destination. According to the assumed locations of ground users and the measurements of the Rician factor of the U2G channels given in~\cite{Survey_channel1, Measure_channel2}, the Rician factors range from $2$ to $12$ in the simulations. Besides, the path loss exponent in our considered U2G channel is set as $\alpha = 2.5$. On the other hand, the UAV and/or the destination are equipped with more antennas than the eavesdroppers in order to achieve positive secrecy rate. Thereafter, the number of antennas at the destination is set to $N_0 = 4$ and the number of antennas at the eavesdroppers is set to $N_1 = \ldots = N_T = 2$, while the number of antennas at the UAV is set to $K = 2$, $4$, or $6$, respectively. At the UAV, the antenna elements are circularly arranged with a radius $10\lambda$, where $\lambda = 0.06$ meter is the wavelength of the 5~GHz carrier frequency. The antennas at each receiver are linearly arranged with $2\lambda$ inter-antenna spacing and randomly orientated.

Fig.~\ref{figRateIter} illustrates the iterations of the secrecy rate achieved at a fixed UAV location by using Algorithm~\ref{alg_power}, when the number of transmit antennas at the UAV is $K = 4$ and $6$. In both cases, we assume that the UAV is placed at three different locations, i.e., $\bm{p}_u = [0,0,10]^\mathrm{T}$, $[0,5,10]^{\mathrm{T}}$, and $[0,10,10]^\mathrm{T}$, respectively, while the destination is located at $\bm{p}_0 = [20,0,0]^\mathrm{T}$, and the eavesdropper is at $\bm{p}_1 = [10,0,0]^{\mathrm{T}}$. In all the cases considered, the convergence of the secrecy rate can be achieved in 4 to 6 iterations. Interestingly, the rate of convergence is faster when $K = 6$, which suggests that secrecy can be realized easier, when the channel's spatial DoF is higher. Additionally, since the LoS component is topology-dependent, the location of the UAV also has a significant impact on the secrecy rate. In particular, when the three nodes are aligned along a horizontal line, the legitimate and the eavesdropping channels are highly correlated, thus resulting in the lowest secrecy rate shown by the solid lines in Fig.~\ref{figRateIter}. As the UAV moves upwards, the two channels gradually become decorrelated and the secrecy rate eventually achieves $60\%$ improvement when $K = 4$, and $20\%$ improvement when $K = 6$.

\begin{figure}[t!]
	\centering
	\includegraphics[width=2.8in]{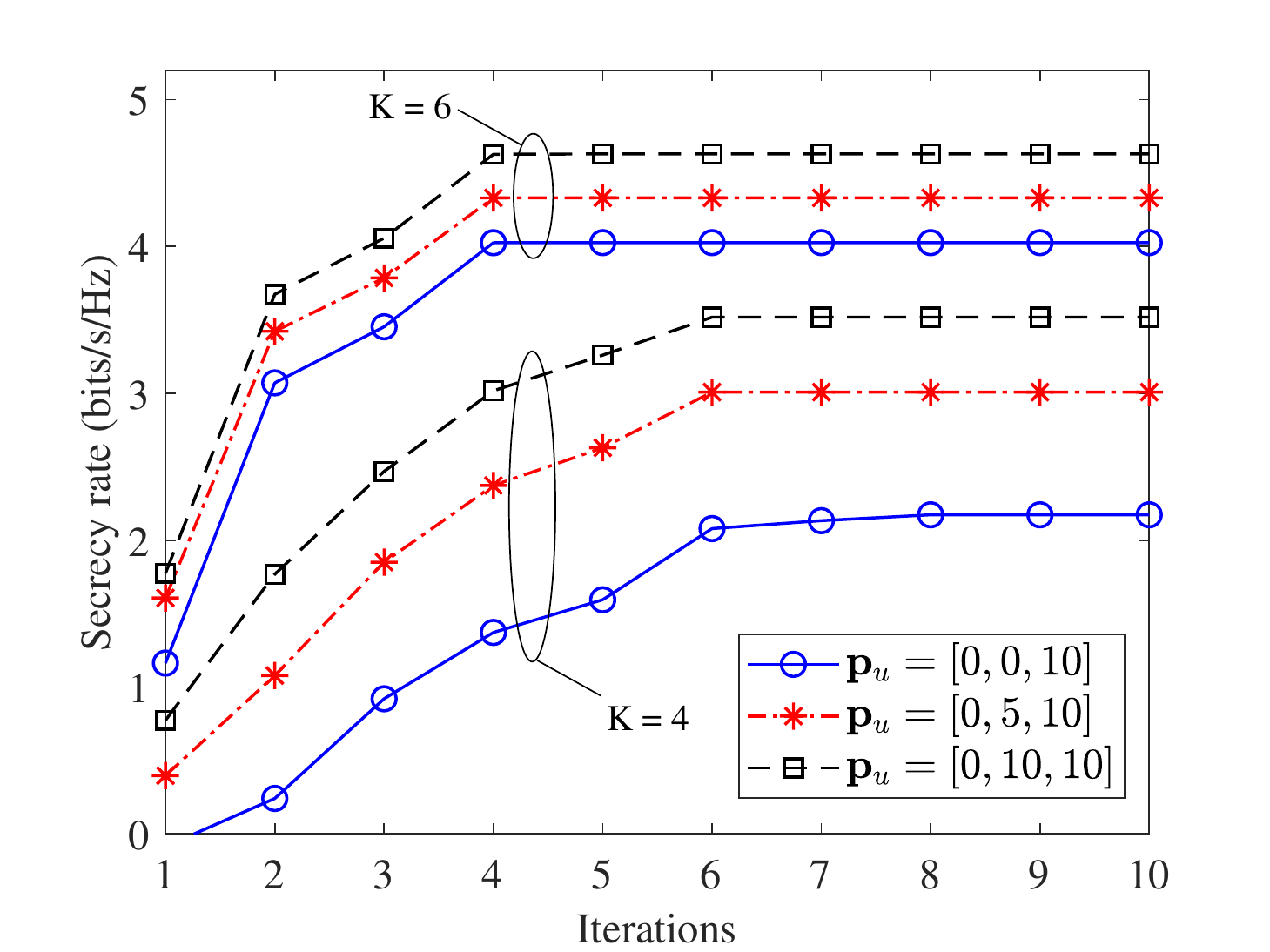}
	\caption{Iterative optimization of the power allocation $\bm{\psi}$ when the UAV is deployed at $\bm{p}_u = [0,0,10]^\mathrm{T}$, $\bm{p}_u = [0,5,10]^\mathrm{T}$, or $\bm{p}_u = [0, 10, 10]^{\mathrm{T}}$, assuming the number of transmit antennas $K = 4$ or $6$.}
	\label{figRateIter}
\end{figure}

Fig.~\ref{figRateTraj_singleEve} shows the iterative updates of the UAV locations until reaching the maximum distance away from its initial dispatch location, as well as the resultant secrecy rates in presence of a single eavesdropper. The initial location of the UAV is set to $\bm{p}_u^{(1)} = [0,0,10]^\mathrm{T}$, the location of the destination is set to $\bm{p}_0 = [20,0,0]^\mathrm{T}$, while the location of the eavesdropper is $\bm{p}_1 = [2, 0, 0]^{\mathrm{T}}$ in Fig.~\ref{figRateTraj_singleEve}~(a) and $\bm{p}_1 = [4,0,0]^{\mathrm{T}}$ in Fig.~\ref{figRateTraj_singleEve}~(b), respectively. As illustrated in the upper sub-figures, for all the cases associated with $K = 2$, $4$, and $6$, by solving the location update subproblem (\ref{eqRsec_pb_sca}), the UAV is capable of avoiding the eavesdropper, while keeping the secrecy rate non-decreasing compared to the previous updates, as shown in the lower sub-figures. When $K = 6$ and the eavesdropper is 2 meters away from the initial location of the UAV, the UAV adjusts its position upwards between the 10-th to 20-th iterations, where the secrecy rates remain approximately constant. When the eavesdropper is 4 meters away from the initial location, the number of iterations required is between 35 to 45. At the optimized locations, compared to the initial location, the secrecy transmissions are indeed facilitated with a positive secrecy rate when $K = 2$ and 4, while they are approximately doubled when $K = 6$.

\begin{figure}[t!]
	\centering
	\subfigure[]{\includegraphics[width=2.8in]{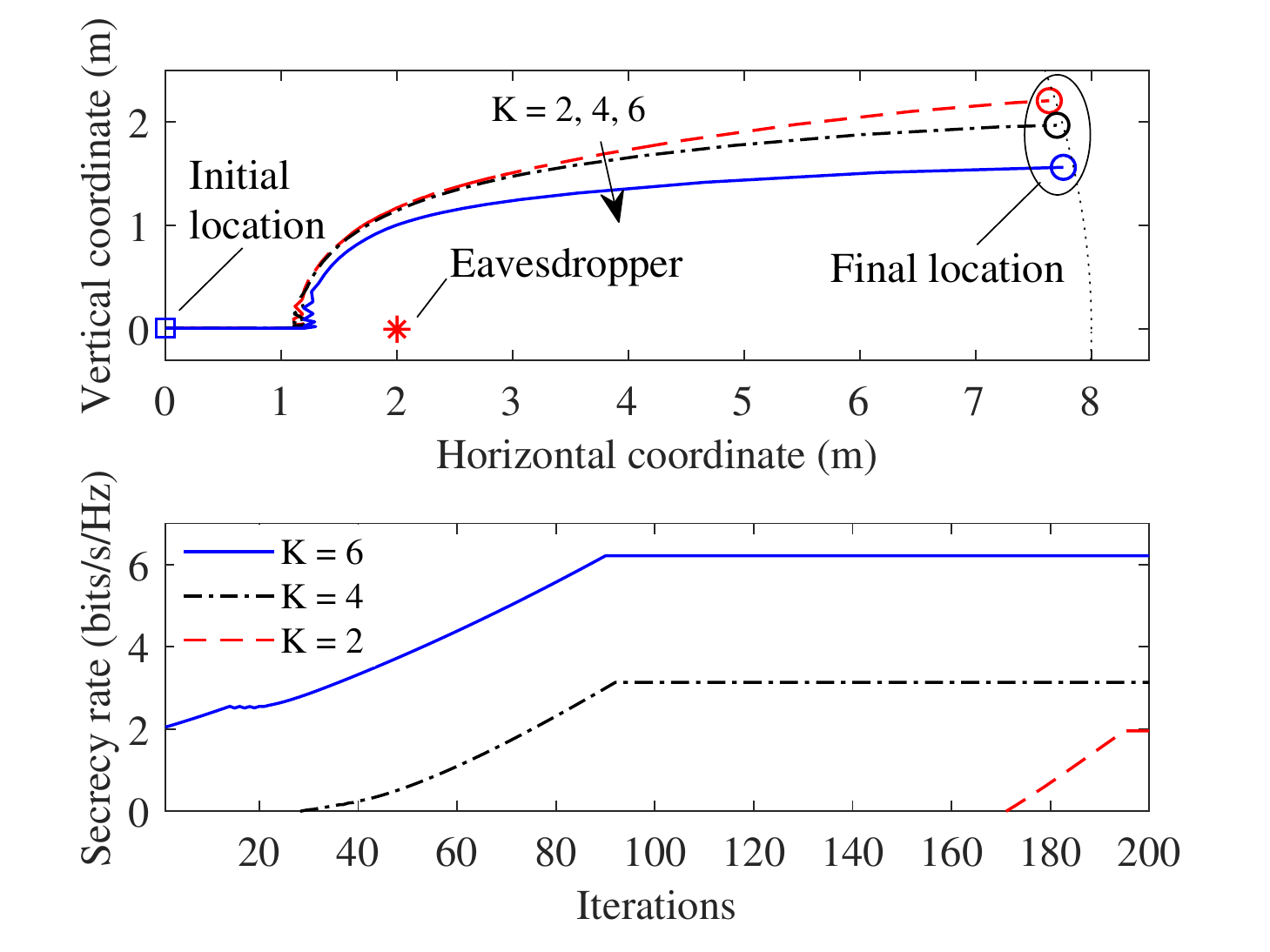}}
	\subfigure[]{\includegraphics[width=2.8in]{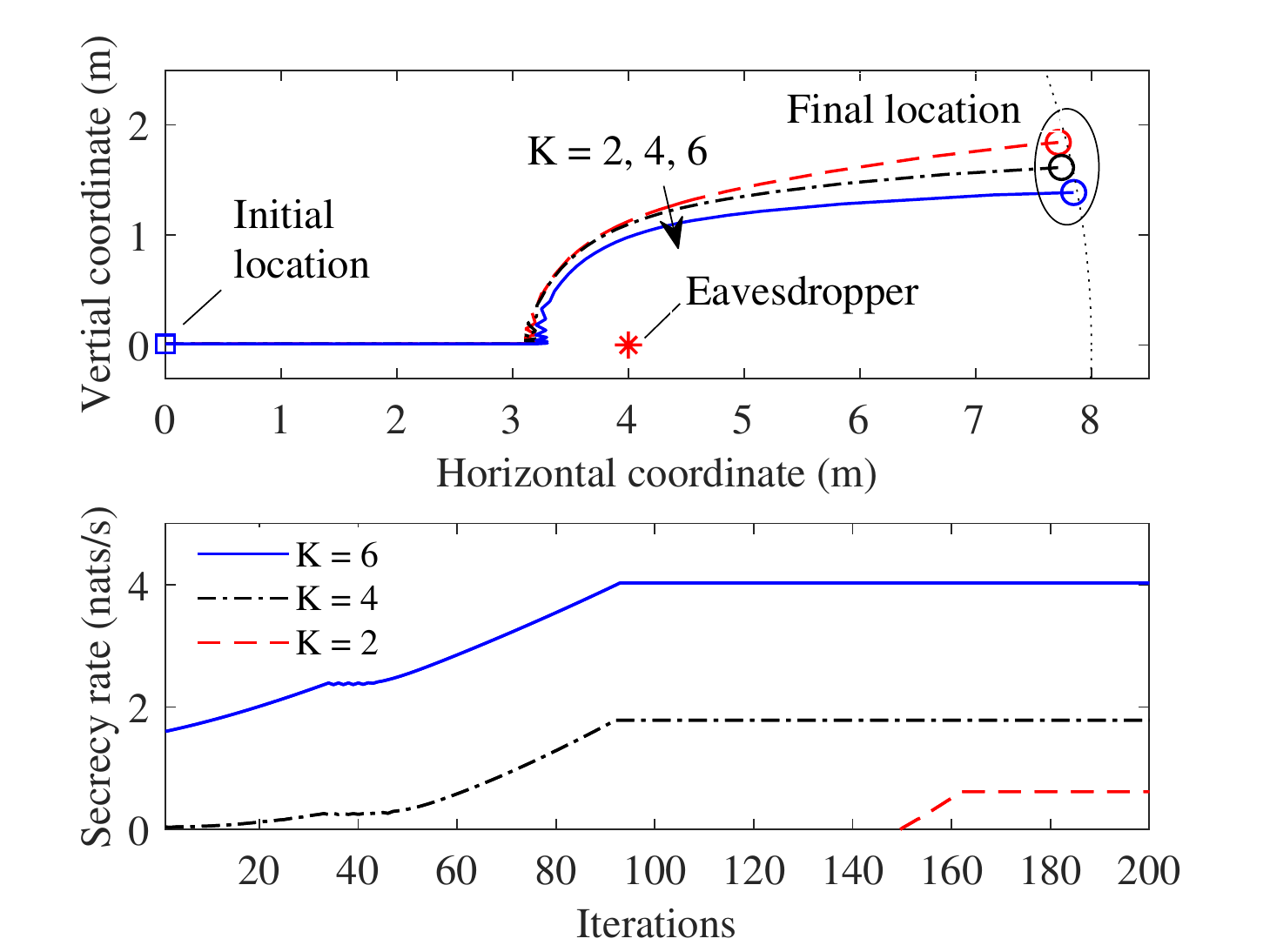}}
	\caption{Trajectory and secrecy rate iterations in presence of a single eavesdropper, assuming the initial location of the UAV as $\bm{p}_u^{(1)} = [0, 0, 10]^{\mathrm{T}}$ and the location of the destination as $\bm{p}_0 = [20, 0, 0]^\mathrm{T}$. The maximum displacement of the UAV is 8 meters. (a) Location of the eavesdropper is $[2, 0, 0]^{\mathrm{T}}$; (b) Location of the eavesdropper is $[4,0,0]^\mathrm{T}$.}
	\label{figRateTraj_singleEve}
\end{figure}

Fig.~\ref{figMultiEve_ed} shows the location sequences and the corresponding secrecy rates optimized by the proposed alternating algorithm, when there are multiple eavesdroppers around the UAV. The optimized results are also compared with those obtained via exhaustive searches~(ES) over all possible combinations of the power allocations and the deployment locations. Similarly to the settings in Fig.~\ref{figRateTraj_singleEve}, the initial dispatch location of the UAV and the location of the destination are set to $\bm{p}_u^{(1)} = [0,0,10]^\mathrm{T}$ and $\bm{p}_0 = [20,0,0]^\mathrm{T}$, respectively. The eavesdroppers are arranged on a $5\times 5$ grid centered at the UAV, while the distance between the adjacent eavesdroppers is 2 meters. It is clear from Fig.~\ref{figMultiEve_ed}~(a) that by applying Algorithm~2, the UAV becomes capable of passing through the grid of eavesdroppers and eventually reaching the maximum displacement distance from the initial dispatch location. Moreover, although the proposed alternating optimization method can only obtain the sub-optimal result, the optimized locations are very close to the ES locations, especially in the case of $K = 2$. The secrecy rates at the initial, the optimized, and the ES locations are compared in Fig.~\ref{figMultiEve_ed}~(b), where the UAV having $K = 6$ antennas achieves more than a factor 5 secrecy rate improvement after the joint optimization of the transmitter and the location. Additionally, it is concluded that the proposed method approaches the optimal secrecy rate.

\begin{figure}[htbp]
		\centering
		\subfigure[]{\includegraphics[width=2.8in]{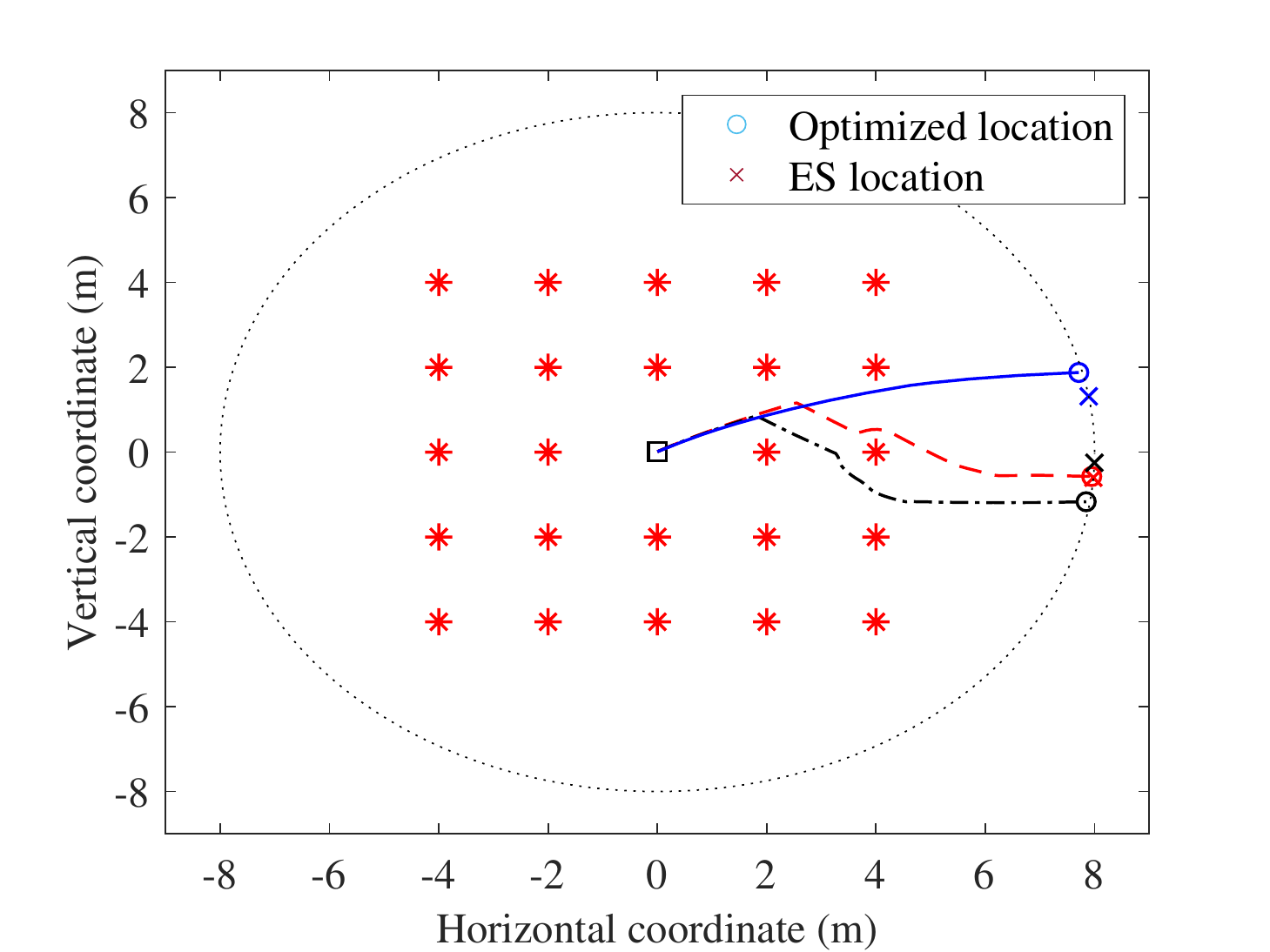}}
		\subfigure[]{\includegraphics[width=2.8in]{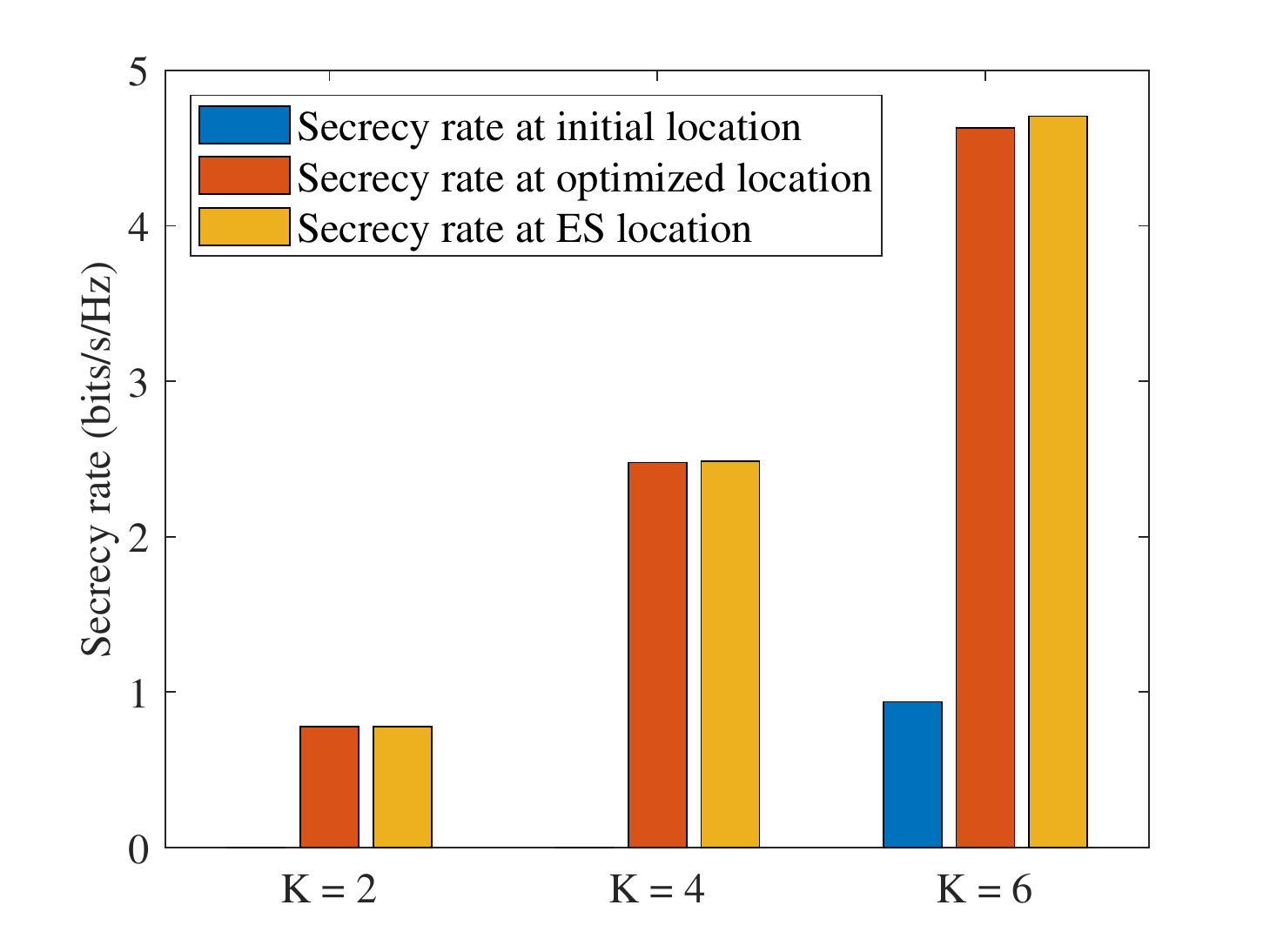}}
		\caption{Trajectory and secrecy rate iterations in presence of multiple eavesdroppers, assuming the initial location of the UAV as $\bm{p}_u^{(1)} = [0, 0, 10]^{\mathrm{T}}$ and the location of the destination as $\bm{p}_0 = [20, 0, 0]^\mathrm{T}$. The maximum displacement of the UAV is 8 meters. (a) Trajectories and the searched location of UAV with $K = 2$, $K = 4$, and $K = 6$; (b) Secrecy rate at the initial, optimized and the ES locations.}\label{figMultiEve_ed}
	\end{figure}

Fig.~\ref{figMultiEveCDF_ed} shows the ECDFs of the instantaneous secrecy rate in the presence of multiple eavesdroppers around the UAV. It is evaluated under the random fading of the propagation channels by using the UAV location and the power allocation obtained by the proposed alternating optimization method. The ECDFs are also compared with those achieved by the ES. Other system configurations show similar trends to those in Fig.~\ref{figRateTraj_singleEve}. Results show that the proposed method achieves similar overall distribution of the secrecy rate as those achieved by the ES. When $K = 4$, the ECDF of the instantaneous secrecy rate obtained by the two approaches nearly coincide. Moreover, it is observed that the fluctuation of the instantaneous secrecy rate caused by the random fading of the propagation channel ranges from $-2$ to $2$ bits/s/Hz around the average secrecy rate. In addition, the secrecy outage can be completely eliminated by the proposed alternating optimization under the antenna configurations of $K = 4$ and $K = 6$.

\begin{figure}[htbp]
		\centering
		{\includegraphics[width=2.8in]{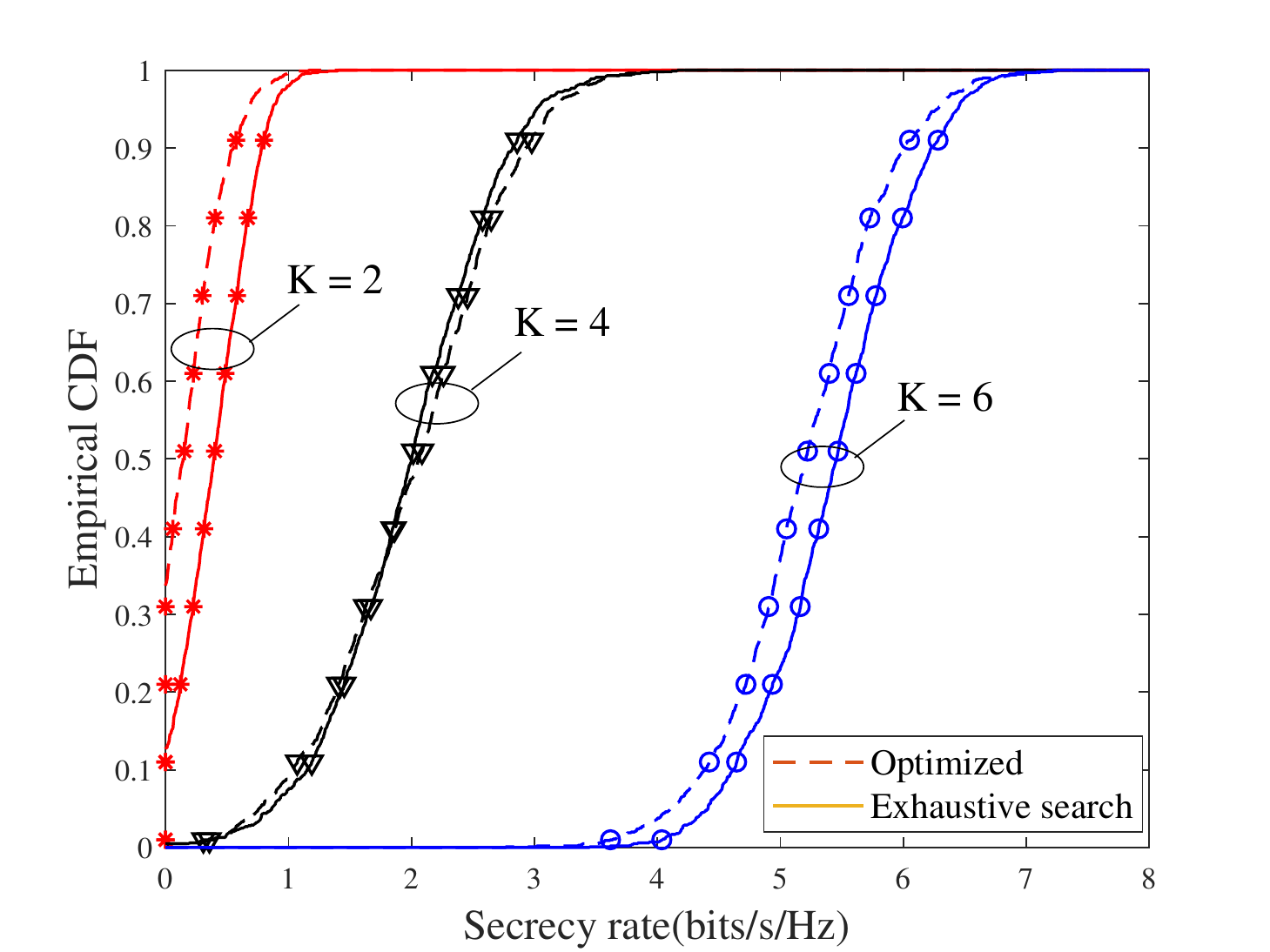}}
		\caption{Empirical cumulative distribution function of instantaneous secrecy rate achieved by the optimized method and by the ES method when there are multiple eavesdroppers around the UAV, assuming the initial location of the UAV as $\bm{p}_u^{(1)} = [0, 0, 10]^{\mathrm{T}}$ and the location of the destination as $\bm{p}_0 = [20, 0, 0]^\mathrm{T}$. The maximum displacement of the UAV is 8 meters. The number of transmit antennas K = 2, 4, 6.}\label{figMultiEveCDF_ed}
\end{figure}

\section{Conclusions}\label{secConclude}

Compared to the terrestrial network establishment, UAV-mounted network nodes have the advantage of flexible and prompt deployment, as well as improved channel conditions. While being readily relocated for security performance improvement, the laser-charging UAVs further avoid frequent recharges and service interruptions. In this paper, we exploited both the on-board antenna array and the maneuverability of the UAVs in secure UAV-to-ground communications, where the multi-antenna transceiver and the UAV deployment are jointly designed while satisfying the secrecy constraint. For this setup, we  first obtained a closed-form expression of the secrecy rate of the UAV-to-ground MIMO wiretap channels by using random matrix theory, which facilitates systematic transceiver and location optimization. Then, an alternating optimization procedure was formulated for iteratively updating the transmit vector and the UAV location. The updates of both optimization variables can be arranged by optimizing the affine representations of the secrecy rate, which are convex and can therefore be solved by standard optimization toolkits. Our results showed that, by using the proposed alternating optimization, the secrecy rate attained increases monotonically along the optimized trajectory between the initial dispatch location and the final deployment location. In particular, when the UAV has the same number of antennas as the eavesdroppers, secure communication with positive secrecy rate can be achieved at the optimized location. When the UAV has more antennas, the proposed algorithm attains substantial secrecy improvements. Thus, the location optimization together with transceiver optimization is of significant importance, which allows the UAV to avoid secrecy outages.

\appendices
\section{Proof of Proposition~\ref{propRtau}}\label{appxRtau}
By denoting $\mb{\Gamma} = \bar{\gamma}/z\mb{\Psi}$ and $\gamma_i = \bar{\gamma}/z\,\psi_i$, $1\le i\le r$, $R_U(\bm{\psi},\bm{p}_u)$ in (\ref{eqRtau_appx}) can be calculated by the following matrix integral:
\begin{align}
	\log \frac{\int \left|\mb{I}_K + \mb{X}^\dagger\mb{X}\mb{T}\mb{\Gamma}\mb{T}^\dagger\right| e^{-\mr{Tr}\left[\left(\mb{X}-\bar{\mb{H}}_{d}\right)\left(\mb{X}-\bar{\mb{H}}_{d}\right)^\dagger\right]} \mr{d}\mb{T}\mr{d}\mb{X}}{\int_{\mc{M}_{N,K}} e^{-\mr{Tr}\left[\left(\mb{X}-\bar{\mb{H}}_d\right)\left(\mb{X}-\bar{\mb{H}}_d\right)^\dagger\right]} \mr{d}\mb{X} },\label{eqR_appx1}
\end{align}
where the numerator of (\ref{eqR_appx1}) integrates over $\mb{X}\in\mc{M}_{N,K}$, the space of $N\times K$ complex matrices, and over $\mb{T}\in\mc{U}_K$, the normalized Haar measure on the unitary group. The integrand $\exp\left[-\mr{Tr}\left\{\left(\mb{X}-\bar{\mb{H}}_d\right)\left(\mb{X}-\bar{\mb{H}}_d\right)^\dagger\right\}\right]\mathrm{d}\mb{X}$ defines the probability measure of $\mb{X}$. The denominator of (\ref{eqR_appx1}) normalizes the right-hand-side of (\ref{eqR_appx1}).

\textbf{The case when $N\ge K$}: Denote the SVD of $\mb{X}$ in (\ref{eqR_appx1}) as $\mb{X} = \mb{U}\mb{\Lambda}^{1/2}\mb{V}^\dagger$, where $\mb{\Lambda} = \mr{diag}(\bm{\lambda})$ and $\bm{\lambda} = [\lambda_1,\ldots,\lambda_K]^{\mr{T}}$ are the eigenvalues of $\mb{X}^\dagger\mb{X}$. When $N\ge K$, we can use change-of-variables in terms of $\mb{U}$, $\mb{\Lambda}$, and $\mb{V}$, to arrive at $\mr{d}\mb{X} = \Delta_K(\bm{\lambda})^2\prod_{i=1}^K \lambda_i^{N-K}\mr{d}\mb{U}\mr{d}\mb{\bm{\lambda}}\mr{d}\mb{V}$, where $\Delta_K(\bm{\lambda}) = |\mb{V}_K(\bm{\lambda})|$ denotes the determinant of the Vandermonde matrix. Substituting $\mb{X} = \mb{U}\mb{\Lambda}^{1/2}\mb{V}^\dagger$ into (\ref{eqR_appx1}), $R_U(\bm{\psi},\bm{p}_u)$ can be rewritten as
\begin{align}
    \log\frac{\int_{[0,\infty)^K} \mc{I}_1(\mb{\Lambda},\mb{\Gamma}) \mc{I}_2(\mb{\Lambda}) \Delta_K(\bm{\lambda})^2\prod_{i=1}^K \lambda_i^{N-K} e^{-\lambda_i}  \mr{d}\bm{\lambda}}   {\int_{[0,\infty)^K} \mc{I}_2(\mb{\Lambda}) \Delta_K(\bm{\Lambda})^2\prod_{i=1}^K \lambda_i^{N-K} e^{-\lambda_i} \mr{d}\bm{\lambda}},\label{eqR_appx2}
\end{align}
where $\mc{I}_1(\mb{\Lambda},\mb{\Gamma}) = \int \left|\mb{I} + \mb{\Lambda}\mb{T}\mb{\Gamma}\mb{T}^\dagger\right| \mr{d}\mb{T}$ and $\mc{I}_2(\mb{\Lambda}) = \int e^{\mr{Tr}\left\{\bar{\mb{H}}_d^\dagger\mb{U}\mb{\Lambda}^{1/2}\mb{V}^\dagger\right\}+\mr{Tr}\left\{\mb{V}\mb{\Lambda}^{\dagger/2}\mb{U}^\dagger\bar{\mb{H}}_d\right\}} \mr{d}\mb{V}  \mr{d}\mb{U}$. Since $\gamma_1\ge\cdots\ge\gamma_r>0$ and $\gamma_{r+1}=\cdots\gamma_K=0$, according to \cite[Eq. (38)]{ZhengTCOM2019}, the integral $\mc{I}_1(\mb{\Lambda},\mb{\Gamma})$ is given by
\begin{align}
    \mc{I}_1(\mb{\Lambda},\mb{\Gamma}) &= \prod_{j=K-r}^{K-1}\!\!\frac{\Gamma(K+1-j)\Gamma(j+1)}{\Gamma(K+1)}\nonumber\\
    &\qquad\times\frac{\mc{D}_{1}^{(1)}(\bm{\lambda})}{\Delta_K(\bm{\lambda})\Delta_r(\bm{\gamma})\prod_{i=1}^r\gamma_i^{K-r}},\label{eqI1}
\end{align}
where $\mc{D}_{1}^{(1)}(\bm{\lambda}) = \left|\left\{\lambda_i^{j-1}\right\}_{K\times(K-r)} \left\{(1+\lambda_i\gamma_j)^{K}\right\}_{K\times r}\right|$.

To obtain $\mc{I}_2(\mb{\Gamma})$, we first assume that the matrix $\bar{\mathbf{H}}_d^\dagger\bar{\mathbf{H}}_d$ is of full-rank, i.e., $\omega_i>0$, $1\le i\le K$. The general expression of $\mc{I}_2(\mb{\Gamma})$ is then obtained by setting the corresponding $\omega_i$ to zero. According to \cite[Eq. (24)]{ChinthaJMP2008}, the expression of $\mc{I}_2(\mb{\Lambda})$ is given by
\begin{align}
    \mc{I}_2(\mb{\Lambda}) = & \frac{\prod_{i=1}^K (N-i)!(K-i)!}{\Delta_K(\bm{\lambda})\Delta_q(\bm{\omega})}\nonumber\\
    &\qquad \times \left|\left\{\frac{I_{N-K}(2\sqrt{\omega_j\lambda_i})}{(\omega_j\lambda_i)^{(N-K)/2}}\right\}_{K\times K}\right|,\label{eqI2}
\end{align}
where $I_{N-K}(x) = \sum_{k=0}^\infty\frac{1}{\Gamma(N-K+1+k)k!}(\frac{x}{2})^{2k+N-K}$ denotes the modified Bessel function. If the rank of $\bar{\mb{H}}_d^\dagger\bar{\mb{H}}_d$ is $0<q<K$, $\mc{I}_2(\mb{\Lambda})$ is obtained by applying \cite[Lemma 1]{ZhengTCOM2019} to (\ref{eqI2}), which drives the $K-q$ eigenvalues $\omega_{q+1},\ldots,\omega_K$ approaching zero. The corresponding $\mc{I}_2(\mb{\Lambda})$ is then proportional to
\begin{align}
\mc{I}_2(\mb{\Lambda}) \sim \frac{\mc{D}_{2}^{(1)}(\bm{\lambda})}{\Delta_K(\bm{\lambda})\Delta_q(\bm{\omega})\prod_{i=1}^q\omega_i^{K-q}},\label{eqI2_D21}
\end{align}
where $\mc{D}_{2}^{(1)}(\bm{\lambda})$ is given by
\begin{align*}
	 \left|\left\{\lambda_i^{j-1}\right\}_{K\times(K-q)} \left\{ \pFq{0}{1}{-}{N-K+1}{\omega_j\lambda_i} \right\}_{K\times q}\right|.
\end{align*}

Substituting $\mc{I}_1(\mb{\Lambda},\mb{\Gamma})$ and $\mc{I}_2(\mb{\Lambda})$ into (\ref{eqR_appx2}), $R_U(\bm{\psi},\bm{p}_u)$ becomes
\begin{align}
	R_U(\bm{\psi},\bm{p}_u) &= \log\frac{\mc{F}_{1}^{(1)}}{\mc{F}_{2}^{(1)}} +\log\frac{\prod_{i=1}^r\gamma_i^{r-K}}{\Delta_r(\bm{\gamma})} \nonumber \\
&\qquad + \sum_{j=K-r}^{K-1}\log\frac{\Gamma(K+1-j)\Gamma(j+1)}{\Gamma(K+1)},\label{eqR_appx3}
\end{align}
where $\mc{F}_{1}^{(1)} = \int \mc{D}_{1}^{(1)}(\bm{\lambda}) \mc{D}_{2}^{(1)}(\bm{\lambda}) \prod_{i=1}^K \lambda_i^{N-K}e^{-\lambda_i} \mr{d}\bm{\lambda}$ and $\mc{F}_{2}^{(1)} = \int \Delta_K(\bm{\lambda}) \mc{D}_{2}^{(1)}(\bm{\lambda}) \prod_{i=1}^K \lambda_i^{N-K}e^{-\lambda_i} \mr{d}\bm{\lambda}$. Applying the generalized Andr\'{e}ief integral \cite[Lemma 2]{ZhengTCOM2019}, $\mc{F}_{1}^{(1)}$ is calculated as
\begin{align}
    K!\left|
        \begin{array}{ll}
            \left\{ a_{j,i} \right\}_{(K-q)\times (K-r)} &
            \left\{ \mc{J}_i^{(1)}(\gamma_j)\right\}_{(K-q)\times r}\\
            \left\{ b_{j,i}  \right\}_{q\times (K-r)} &
        \left\{ \mc{K}_i^{(1)}(\gamma_j)\right\}_{q\times r}
        \end{array}
        \right|,\label{eqR_num1}
\end{align}
where $\mc{J}_i^{(1)}(\gamma_j)$ and $\mc{K}_i^{(1)}(\gamma_j)$ are given by
\begin{align}
\mc{J}_i^{(1)}(\gamma_j) &= \int_0^\infty \lambda^{N-K+i-1}e^{-\lambda}(1+\gamma_j\lambda)^{K}\mr{d}\lambda \nonumber \\
&= \sum_{n=0}^K {K\choose n} a_{n+1,i} \gamma_j^n,\\
\mc{K}_i^{(1)}(\gamma_j) &= \int_0^\infty \frac{\lambda^{N-K}e^{-\lambda}}{(1+\gamma_j\lambda)^{-K}} \pFq{0}{1}{-}{N-K+1}{\omega_i\lambda} \mr{d}\lambda \nonumber \\
&= \sum_{n=0}^K {K\choose n} b_{n+1,i} \gamma_j^n.
\end{align}

Similarly, the integral $\mc{F}_{2}^{(1)}$ in (\ref{eqR_appx3}) is calculated as $\mc{F}_{2}^{(1)} = K!\left|
    \begin{array}{ll}
    	\left\{ a_{i,j} \right\}_{K\times (K-q)} &
    	\left\{ b_{i,j} \right\}_{K\times q}
    \end{array}
    \right|$. Substituting $\mc{F}_{1}^{(1)}$ into (\ref{eqR_appx3}) and sequentially absorbing the $(\prod_{i=1}^r\gamma_i)^{r-K}$ term in (\ref{eqR_appx3}) into the determinant of $\mc{F}_{1}^{(1)}$, the constant terms and the first $(K-r)$ lower-power terms of the polynomials $\mc{J}_i^{(1)}(\gamma_j)$ and $\mc{K}_i^{(1)}(\gamma_j)$ are canceled by the corresponding $a_{i,j}$ and $b_{i,j}$ in the left sub-matrices of (\ref{eqR_num1}) due to the multilinearity of the matrix determinant. Finally, substituting $\mc{F}_{2}^{(1)}$ into results in (\ref{eqR_appx3}) attains (\ref{eqRtau_4}) when $N\ge K$.

\textbf{The case when $K>N$}: Following similar procedures as (\ref{eqR_appx1})-(\ref{eqR_appx2}), $R_U(\bm{\psi},\bm{p}_u)$ is obtained as
\begin{align}
	&R_U(\bm{\psi},\bm{p}_u) =\nonumber \\
    & \log\frac{\int_{[0,\infty)^N} \mc{I}_1(\mb{\Lambda},\mb{\Gamma}) \mc{I}_2(\mb{\Lambda}) \Delta_N(\bm{\lambda})^2\prod_{i=1}^N \lambda_i^{K-N} e^{-\lambda_i}  \mr{d}\bm{\lambda}}   {\int_{[0,\infty)^N} \mc{I}_2(\mb{\Lambda}) \Delta_N(\bm{\lambda})^2\prod_{i=1}^N \lambda_i^{K-N} e^{-\lambda_i} \mr{d}\bm{\lambda}}.\label{eqPhi5}
\end{align}
According to \cite{ZhengTCOM2019}, when $K>N\ge r$, we obtain
\begin{align}
&\mc{I}_1(\mb{\Lambda},\mb{\Gamma}) =\nonumber \\
& \prod_{j=K-r}^{K-1}\frac{\Gamma(K+1-j)\Gamma(j+1)}{\Gamma(N+1)\Gamma(K-N+1)}\frac{\prod_{j=1}^r\gamma_j^{r-N}}{\Delta_N(\bm{\lambda})\Delta_r(\bm{\gamma})}\mc{D}_{1}^{(2)}(\bm{\lambda}),\label{eqI1_2}
\end{align}
where $\mc{D}_{1}^{(2)}(\bm{\lambda})$ is given by
\begin{align*}
	\left|
	\begin{array}{cc}
		\left\{\lambda_i^{j-1}\right\}_{N\times (N-r)} & \left\{\pFq{2}{1}{1,-N}{K-N+1}{-\gamma_j \lambda}\right\}_{N\times r}
	\end{array}
	\right|.
\end{align*}
Given $\mc{I}_1(\mb{\Lambda},\mb{\Gamma})$ in (\ref{eqI1_2}) and $\mc{I}_2(\mb{\Lambda})$ in (\ref{eqI2_D21}) while interchanging $N$ and $K$, $R_U(\bm{\psi},\bm{p}_u)$ in (\ref{eqPhi5}) is obtained for $K>N\ge r$ as
\begin{align}
    R_U(\bm{\psi},\bm{p}_u) =& \log\frac{\mc{F}_{1}^{(2)}}{\mc{F}_{2}^{(2)}} +\log\frac{\prod_{j=1}^r\gamma_j^{r-N}}{\Delta_r(\bm{\gamma})}\nonumber \\
    & + \sum_{j=K-r}^{K-1} \log\frac{\Gamma(K+1-j)\Gamma(j+1)}{\Gamma(N+1)\Gamma(K-N+1)},\label{eqR_appx4}
\end{align}
where $\mc{F}_{1}^{(2)} = \int \mc{D}_{1}^{(2)}(\bm{\lambda})\mc{D}_{2}^{(2)}(\bm{\lambda}) \prod_{i=1}^N \lambda_i^{K-N}e^{-\lambda_i} \mr{d}\bm{\lambda}$ and $\mc{F}_{2}^{(2)} = \int \mc{D}_{2}^{(2)}(\bm{\lambda}) \Delta_N(\bm{\lambda}) \prod_{i=1}^N \lambda_i^{K-N} e^{-\lambda_i} \mr{d}\bm{\lambda}$.

Similar to (\ref{eqR_num1}), the integral $\mc{F}_{1}^{(2)}$ in (\ref{eqR_appx4}) can be calculated as
\begin{align}
    N!\left|
        \begin{array}{ll}
            \left\{ a_{j,i} \right\}_{(N-q)\times (N-r)} &
            \left\{ \mc{J}_i^{(2)}(\gamma_j) \right\}_{(N-q)\times r} \\
            \left\{ b_{j,i}  \right\}_{q\times (N-r)} &
        \left\{ \mc{K}_i^{(2)}(\gamma_j) \right\}_{q\times r}
        \end{array}
        \right|,\label{eqR_num2}
\end{align}
where $\mc{J}_i^{(2)}(\gamma_j)$ and $\mc{K}_i^{(2)}(\gamma_j)$ are given by
\begin{align}
\mc{J}_i^{(2)}(\gamma_j) &= \int_0^\infty \lambda^{K-N+i-1} e^{-\lambda} \pFq{2}{1}{1,-N}{K-N+1}{-\gamma_j \lambda} \mr{d}\lambda \nonumber \\
&= \sum_{n = 0}^M {K\choose M-n} a_{n+1,i} \gamma_j^n,\label{eqJ_2}
\end{align}
\begin{align}
\mc{K}_i^{(2)}(\gamma_j) &= \int_0^\infty \frac{e^{-\lambda}}{\lambda^{N-K}} \pFq{0}{1}{-}{K-N+1}{\omega_i\lambda} \nonumber \\
&\qquad\qquad\  \times \pFq{2}{1}{1,-N}{K-N+1}{-\gamma_j \lambda} \mr{d}\lambda \nonumber \\
& = \sum_{n = 0}^M {K\choose M-n} b_{n+1,i}\gamma_j^{n}.\label{eqK_2}
\end{align}

The integral $\mc{F}_{2}^{(2)}$ in (\ref{eqR_appx4}) can be calculated as $\mc{F}_{2}^{(2)} = N!\left|
        \begin{array}{ll}
            \left\{ a_{i,j} \right\}_{(N-q)\times N} &
            \left\{ b_{i,j} \right\}_{q\times N}
        \end{array}
        \right|$. Then, we follow similar procedures by absorbing the $\prod_{j=1}^r\gamma_j^{r-N}$ term in (\ref{eqR_appx4}) into the determinant of $\mc{F}_{1}^{(2)}$, which leads to (\ref{eqRtau_4}), when $K>N\ge r$.

When $K>r>N$, we obtain
\begin{align}
	&\mc{I}_1(\mb{\Lambda},\mb{\Gamma})=\nonumber\\
	&\times\prod_{j=r-N}^{r-1}\frac{\Gamma(r+1-j)\Gamma(K+j-r+1)}{\Gamma(N+1)\Gamma(K-N+1)}\frac{\mc{D}_{1}^{(3)}(\bm{\lambda})}{\Delta_N(\bm{\lambda})\Delta_r(\bm{\gamma})},\label{eqI1_3}
\end{align}
where $\mc{D}_{1}^{(3)}(\bm{\lambda})$ is given by
\begin{align*}
	\left|
	\begin{array}{cc}
		 \left\{\gamma_i^{j-1}\right\}_{r\times (r-N)}  & \left\{\gamma_i^{r-N}\pFq{2}{1}{1,-N}{K-N+1}{-\lambda_j\gamma_i} \right\}_{r\times N}
	\end{array}
	\right|.
\end{align*}
By substituting (\ref{eqI1_3}) and $\mc{I}_2(\mb{\Lambda})$ into (\ref{eqPhi5}), $R_U(\bm{\psi},\bm{p}_u)$ is obtained as
\!\!\begin{align}
    &R_U(\bm{\psi},\bm{p}_u) = \log\frac{\mc{F}_1^{(3)}}{\mc{F}_2^{(2)}} - \log\Delta_r(\bm{\gamma}) \nonumber \\
    &\qquad + \sum_{j=r-N}^{r-1}\log\frac{\Gamma(r+1-j)\Gamma(K+j-r+1)}{\Gamma(N+1)\Gamma(K-N+1)} ,\label{eqR_appx5}
\end{align}
where $\mc{F}_1^{(3)} = \int \mc{D}_{1}^{(3)}(\bm{\lambda})\mc{D}_{2}^{(2)}(\bm{\lambda})\prod_{i=1}^N \lambda_i^{K-N} e^{-\lambda_i}\mr{d}\bm{\lambda}$. The integral $\mc{F}_{1}^{(3)}$ in (\ref{eqR_appx5}) is calculated as
\begin{align*}
	N! \left|
	\begin{array}{ccc}
		\left\{ \gamma_j^{i-1} \right\}_{(r-N)\times r} \\
		\left\{ \gamma_j^{r-N}\mc{J}_i^{(2)}(\gamma_j) \right\}_{(N-q)\times r} \\
		\left\{ \gamma_j^{r-N}\mc{K}_i^{(2)}(\gamma_j) \right\}_{q\times r}
	\end{array}
	\right|,
\end{align*}
where $\mc{J}_i^{(2)}(\gamma_j)$ and $\mc{K}_i^{(2)}(\gamma_j)$ are given by (\ref{eqJ_2}) and (\ref{eqK_2}), respectively. Finally, by substituting $\mc{F}_1^{(3)}$ and $\mc{F}_2^{(2)}$ into (\ref{eqR_appx5}), we obtain $R_U(\bm{\psi},\bm{p}_u)$ in (\ref{eqRtau_4}), when $K>r>N$.


\ifCLASSOPTIONcaptionsoff
  \newpage
\fi



%

\bibliographystyle{IEEEtran}
\bibliography{IEEEabrv,Ref}

\end{document}